\documentclass[11pt]{article}
\newcommand{\LIPICS}[1]{}\newcommand{\NORMAL}[1]{#1}

\NORMAL{

\usepackage{times,fullpage,amsmath,amsthm,amssymb,graphicx,color,hyperref}

\title{Convex Polygon Containment:\\ Improving Quadratic to Near Linear Time}
\author{Timothy M. Chan\thanks{Department of Computer Science, University of Illinois at Urbana-Champaign (tmc@illinois.edu).
Work supported in part by NSF Grant CCF-2224271.}  \and 
Isaac M. Hair\thanks{Department of Computer Science, University of California, Santa Barbara (hair@ucsb.edu).}}

\newtheorem{lemma}{Lemma}
\newtheorem{theorem}[lemma]{Theorem}
\newtheorem{corollary}[lemma]{Corollary}

}

\usepackage{todonotes}
\usepackage{float}

\LIPICS{
\usepackage[utf8]{inputenc}

\usepackage{microtype,hyperref,graphicx}
   \hypersetup{%
      breaklinks,%
      colorlinks=true,%
      urlcolor=[rgb]{0.25,0.0,0.0},%
      linkcolor=[rgb]{0.5,0.0,0.0},%
      citecolor=[rgb]{0,0.2,0.445},%
      filecolor=[rgb]{0,0,0.4},
      anchorcolor=[rgb]={0.0,0.1,0.2}%
   }

\title{Convex Polygon Containment:\\ Improving Quadratic to Near Linear Time}

\author{Timothy M. Chan}{Department of Computer Science, University of Illinois at Urbana-Champaign, USA}{tmc@illinois.edu}{https://orcid.org/0000-0002-8093-0675}{Work supported by NSF Grant CCF-2224271.}
\author{Isaac M. Hair}{Department of Computer Science, University of California, Santa Barbara, USA}{hair@ucsb.edu}{https://orcid.org/0000-0001-6992-4488}{}

\titlerunning{Convex Polygon Containment}
\authorrunning{T.\,M. Chan and I.\,M. Hair}

\Copyright{Timothy M. Chan and Isaac M. Hair}



\relatedversiondetails{Full Version}{https://arxiv.org/abs/XXX}

\EventEditors{Wolfgang Mulzer and Jeff M. Phillips}
\EventNoEds{2}
\EventLongTitle{40th International Symposium on Computational Geometry
(SoCG 2024)}
\EventShortTitle{SoCG 2024}
\EventAcronym{SoCG}
\EventYear{2024}
\EventDate{June 11-14, 2024}
\EventLocation{Athens, Greece}
\EventLogo{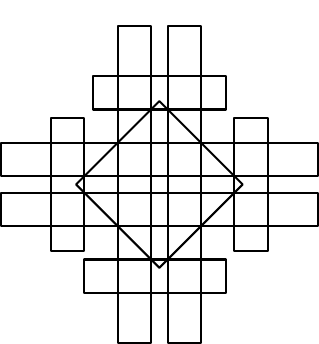}
\SeriesVolume{293}
\ArticleNo{XX}     % <-- This will be filled in by the typesetters

\renewcommand{\paragraph}[1]{\subparagraph*{#1}}

}

\newtheorem{problem}{Problem}

\newcommand{\R}{\mathbb{R}}
\newcommand{\polylog}{\mathop{\rm polylog}}
\newcommand{\eps}{\varepsilon}
\newcommand{\OO}{\widetilde{O}}
\newcommand{\PPP}{{\cal P}}

\newcommand{\LINE}[1]{\overleftrightarrow{#1}}
\newcommand{\ext}[1]{\overleftrightarrow{#1}}
\newcommand{\INT}{I} %{\textsc{interval}}
\newcommand{\MATCH}{M} %{\textsc{matched-edge}}
\newcommand{\IGNORE}[1]{}

\begin{document}

\maketitle

\begin{abstract}
We revisit a standard polygon containment problem: given a convex $k$-gon $P$ and a convex $n$-gon $Q$ in the plane, find a placement of $P$ inside $Q$ under translation and rotation (if it exists), or more generally, find the largest copy of $P$ inside $Q$ under translation, rotation, and scaling.

Previous algorithms by Chazelle (1983), Sharir and Toledo (1994), and Agarwal, Amenta, and Sharir (1998) all required $\Omega(n^2)$ time, even in the simplest $k=3$ case. We present a significantly faster new algorithm for $k=3$ achieving $O(n\polylog n)$ running time. Moreover, we extend the result for general $k$, achieving $O(k^{O(1/\varepsilon)}n^{1+\varepsilon})$ running time for any  $\varepsilon>0$.

Along the way, we also prove a new $O(k^{O(1)}n\polylog n)$ bound on the number of %(combinatorially distinct) 
similar copies of $P$ inside $Q$ that have 4 vertices of $P$ in contact with the boundary of $Q$ (assuming general position input), disproving a conjecture by Agarwal, Amenta, and Sharir (1998).
\end{abstract}

\section{Introduction}
\emph{Polygon containment} problems have been studied since the early years of computational geometry
\cite{AgarwalAS,AgarwalAS99,AvBoissonnat,BarequetHarPeled,Chazelle,ChewKedem,DanielsM97,DickersonS96,EomLeeAhn,KunnemannNusser,LeeEomAhn,Milenkovic96,Milenkovic97,Milenkovic,ORourkeSuriTothPOLYGONS,SharirT}.
In this paper, we focus on two of the most fundamental versions of the problem for convex polygons:
\begin{problem}\label{prob}
Given a convex $k$-gon $P$ and a convex $n$-gon $Q$ in $\R^2$, (i)~find a congruent copy of $P$ inside $Q$ (if it exists);
or more generally, (ii)~find the largest similar copy of $P$ inside~$Q$.
\end{problem}

In a congruent copy, we allow translation and rotation; in a similar copy, we allow translation, rotation,
and scaling.
%(We exclude reflection, but the problem with reflection simply reduces to two instances of the problem without reflection.)
Rotation is what makes the problem challenging, as the corresponding problem without rotation 
can be solved in linear time by a simple reduction to linear programming in 3 variables~\cite{SharirT}.

There were 3 key prior papers on this problem:

\begin{enumerate}
\item
In 1983, Chazelle~\cite{Chazelle} initiated the study of polygon containment
problems and presented an $O(kn^2)$-time algorithm specifically for Problem~\ref{prob}(i).  In particular, an entire section of his
paper was devoted to an $O(n^2)$-time algorithm just for the $k=3$ (triangle) case.

\item 
Sharir and Toledo~\cite{SharirT} (preliminary version in  SoCG'91) applied parametric search~\cite{Megiddo} to reduce
various versions of ``extremal'' polygon containment problems (about finding largest copies) to their corresponding decision problems.
In particular, they described an $O(n^2\log^2n)$-time algorithm for Problem~\ref{prob}(ii) in the $k=3$ case.

\item
In 1998, Agarwal, Amenta, and Sharir~\cite{AgarwalAS} studied Problem~\ref{prob}(ii)
and obtained an $O(kn^2\log n)$-time algorithm. Their approach is to explore the entire solution space.
More precisely, consider a standard 4-parameter representation of similarity transformations~\cite{Baird}: given $(s,t,u,v)\in\R^4$, let $\varphi_{s,t,u,v}:\R^2\rightarrow\R^2$ be the similarity transformation $(x,y)\mapsto 
(sx-ty+u,tx+sy+v)$, which has scaling factor $\sqrt{s^2+t^2}$.  %Let $\sigma=\{(s,t,u,v): s^2+t^2=1\}$.
The region 
\[ \PPP=\{(s,t,u,v)\in\R^4:\ \mbox{$\varphi_{s,t,u,v}(p)\in Q$ for all vertexes $p$ of $P$}\}
\]
describes all feasible solutions and
is an intersection of $O(kn)$ halfspaces in $\R^4$  (since $\varphi_{s,t,u,v}(p)$ is a linear function in the 4
variables $s,t,u,v$ for any fixed point $p$).  The problem is to find a point in $\PPP$ maximizing the convex
function $s^2+t^2$ (the optimum must be located at a vertex). By standard results, a 4-polytope with
$O(kn)$ facets has $O(k^2n^2)$ combinatorial complexity (and can be constructed in $O(k^2n^2)$ time)~\cite{PreparataS}. Agarwal, Amenta, and Sharir improved the combinatorial bound
to $O(kn^2)$ for this particular polytope $\PPP$, enabling them to derive an algorithm with a similar time bound.
\end{enumerate}

Notice that all these previous algorithms have $\Omega(n^2)$ time complexity, even in the triangle ($k=3$) case. (Other quadratic algorithms for $k=3$ have been found, e.g., mostly recently by Lee, Eom, and Ahn~\cite{LeeEomAhn}.)
To explain why, Chazelle \cite{Chazelle} mentioned that there are input convex polygons $Q$ for which the number of different
``stable solutions'' is $\Omega(n^2)$.  (Other authors made similar observations~\cite{LeeEomAhn}.)
More generally, Agarwal, Amenta, and Sharir \cite{AgarwalAS} exhibited a construction of input convex polygons
$P$ and $Q$ for which the polytope $\PPP$ has complexity $\Omega(kn^2)$, matching their
combinatorial upper bound.

Although such combinatorial lower bound results do not technically rule out the possibility of faster algorithms
that find an optimal solution without generating the entire solution space, they indicate that quadratic complexity
is a natural barrier.
%, and this might have also deterred researchers from seeking subquadratic algorithms for this problem.
In general, no techniques are known to maximize a convex function in an intersection of halfspaces in
constant dimensions with worst-case time better than constructing the entire halfspace intersection.

What motivates us to revisit this topic is the  similarity of the $k=3$ problem to the well-known \emph{3SUM problem} (in the sense that  the main case of the triangle problem is about finding a \emph{triple} of vertices/edges of $Q$ in contact with the triangle $P$).
Our initial thought is to apply the exciting recent advances for 3SUM and related problems~\cite{Aronov23,Barba19,Chan18,Gronlund14,Kane18} to design
decision trees with subquadratic height. This would potentially lead to slightly subquadratic algorithms with running time of the form $n^2/\polylog n$.

Although we believe this line of attack can indeed be applied to Problem~\ref{prob}
in the $k=3$ case, the improvement in the time complexity would be tiny, and generalization to $k>3$ is unclear.
Furthermore, the usage of such heavy machinery might seem premature and unjustified since 
the $k=3$ problem has not been shown to be 3SUM-hard.
Barequet and Har-Peled~\cite{BarequetHarPeled} proved that Problem~\ref{prob}(i) for convex polygons is 3SUM-hard when $k=n$ and so has a near-quadratic conditional lower bound, but for $k=n$, the current upper bound is cubic.
More recently, K\"unnemann and Nusser~\cite{KunnemannNusser} have obtained conditional lower bounds for a number of other polygon containment problems, but not in the convex cases.

\paragraph{New results.}
In this paper, we not only truly break the quadratic barrier but also discover a near-linear, $O(n\polylog n)$-time algorithm for Problem~\ref{prob}(ii) in the $k=3$ case!  This represents a substantial improvement over the previous algorithms from
multiple decades earlier, and directly addresses an open problem posed by Agarwal, Amenta, and Sharir \cite{AgarwalAS} asking for an algorithm faster than $\Theta(kn^2)$ time. (We cannot think of too many classical 2D problems in computational geometry of comparable stature
where quadratic/superquadratic time complexity is reduced to near-linear in a single swoop after a long gap. The closest analog is perhaps
Sharir's breakthrough $O(n\polylog n)$-time algorithm for the 2D Euclidean 2-center problem~\cite{Sharir96} that improved a string of
previous $O(n^2\polylog n)$-time algorithms.)

Furthermore, we generalize our approach and obtain an $O(n^{1+\eps})$-time algorithm for Problem~\ref{prob}(ii) for any constant $k>3$, where $\eps>0$ is an arbitrarily
small constant.  For non-constant $k \le n$, the time bound is $O(k^{O(1/\eps)}n^{1+\eps})$ for any choice of (possibly non-constant) $\eps>0$.
(By choosing $\eps=\sqrt{\log k/\log n}$, the bound can be rewritten as $n2^{O(\sqrt{\log k\log n})}$.)
This beats the previous $O(kn^2)$ bound for all $k < n^\alpha$ for some concrete constant $\alpha>0$.

\paragraph{New approach.}
Although Problem~\ref{prob}(ii) reduces to Problem~\ref{prob}(i) by parametric search~\cite{Megiddo,SharirT} (if one does not mind extra logarithmic factors), we actually find it more convenient to solve Problem~\ref{prob}(ii) directly (i.e., find the largest copy).
The optimal solution must belong to one of the following cases, as observed in previous works (by simple direct arguments, or by
recalling that the optimum corresponds to a vertex of the 4-polytope~$\PPP$):

\begin{itemize}
\item {\bf 2-Contact (i.e., 2-Anchor) Case:} 2 distinct vertices of $P$ are in contact with $\partial Q$, both of which are at 2 vertices of $Q$. These 2 vertices of $P$ are called the 2 ``anchor'' vertices.
\item {\bf 3-Contact (i.e., 1-Anchor) Case:} 3 distinct vertices of $P$ are in contact with $\partial Q$, one at a vertex of $Q$ and the other two
on edges of $Q$. The vertex of $P$ placed at a vertex of $Q$ is referred to as the ``anchor'' vertex.
\item {\bf 4-Contact (i.e., No-Anchor) Case:} 4 distinct vertices of $P$ are in contact with $\partial Q$, all on edges of $Q$.
\end{itemize}

For $k=3$, the main case is the 3-contact case, since it turns out that the 2-contact case can be solved in a similar way (and the 4-contact case, of course, does not arise).
The overall strategy is to divide into sub-problems involving different ``arcs'' (i.e., contiguous pieces) of $\partial Q$.
Our key observation is that under certain conditions about the slopes/angles of the input arcs, all 3-contact feasible solutions
may be covered by just a linear number of pairs of sub-edges, due to \emph{monotonicity} arguments---this is despite the fact that the total number of 3-contact solutions may be quadratic.  In such scenarios,
we can search for the best solution by using standard geometric data structuring techniques (concerning
intersections of ellipses, as it turns out).  A simple binary divide-and-conquer reduces to instances where
such conditions are met, resulting in an $O(n\polylog n)$-time algorithm.

% For $k>3$, extending the 3-contact algorithm requires more technical effort (and a slightly increased running time), but
% what appears even more challenging is the 4-contact case. The lack of anchor vertices seems to make
% everything more complicated, especially the needed geometric data structures.
% However, with a different strategy,
% we show that the 4-contact case is actually \emph{easier} in the sense that the total number of 4-contact solutions
% is near-linear in $n$. Thus, we can afford to enumerate all solutions!
% We prove this combinatorial bound by running our $k=3$ algorithm on different triples of vertices of $P$ and then piecing
% this information together via further monotonicity arguments.

For $k>3$, extending the 3-contact algorithm requires more technical effort (and a slightly increased running time), but
what appears even more challenging is the 4-contact case. The lack of anchor vertices seems to make
everything more complicated (including the needed geometric data structures).
However, with a different strategy,
we show surprisingly that the 4-contact case is easier in the sense that the total number of 4-contact feasible solutions
is actually near-linear in $n$, namely, $O(k^4 n\polylog n)$ (assuming general position input).  Thus, we can afford to enumerate them all!
We prove this combinatorial bound by running our $k=3$ algorithm on different triples of vertices of $P$ and then piecing  information together via further interesting monotonicity arguments.

To see how counterintuitive our near-linear combinatorial bound for 4-contact solutions is,
recall that Agarwal, Amenta, and Sharir~\cite{AgarwalAS} proved an $\Omega(kn^2)$ lower bound on the size
of the solution space.  They noted that their construction only lower-bounded the number of 3-contact solutions,
and at the end of their paper, they asked for another construction with $\Omega(kn^2)$ 4-contact solutions. Our proof answers their question in the negative.

% In giving our algorithm for $k > 3$, we prove a new, near-linear $O(k^4n\polylog n)$ combinatorial bound on the number of similar copies of $P$ inside $Q$ that have four vertices of $P$ in contact with $\partial Q$; this directly addresses \emph{yet another} open problem posed by Agarwal, Amenta, and Sharir \cite{AgarwalAS} regarding the worst-case number of 4-contact solutions. Our result is especially surprising in light of the previous constructions \cite{Chazelle, AgarwalAS, LeeEomAhn} that give a quadratic number of combinatorially distinct placements of $P$ within $Q$. As it turns out, these constructions all rely on having just 3 vertices of $P$ in contact with $Q$, so our new result shows that the 4-contact case is actually combinatorially simpler!

% To see how counterintuitive our near-linear time algorithm is,
% recall that Agarwal, Amenta, and Sharir~\cite{AgarwalAS} proved an $\Omega(kn^2)$ lower bound on the size
% of the solution space.  They noted that their construction only lower-bounded the number of 3-contact solutions,
% and at the end of their paper, they posed as an open question whether there is another construction with an $\Omega(kn^2)$ number of 4-contact solutions.

\paragraph{Preliminaries.} 
The \emph{angle} of a line $\LINE{p_1p_2}$, denoted $\theta_{p_1p_2}$, 
refers to the angle measured counterclockwise (ccw) from the $x$-axis to $\LINE{p_1p_2}$.
%is the angle with $\LINE{p_1p_2}$ as its left support and a horizontal line as its right support. 
(Note that $\theta_{p_1p_2}\in [0,\pi)$ and $\theta_{p_2p_1}=\theta_{p_1p_2}$.)
An \emph{arc} $\Gamma$ of a convex polygon $Q$ refers to a contiguous portion of the boundary $\partial Q$ whose supporting lines have
angles in an interval of length $<\pi/3$.
Let $\Lambda(\Gamma)$ (the \emph{angle range} of $\Gamma$) denote the interval containing the angles of all supporting lines of $\Gamma$.
We allow $\Lambda(\Gamma)$ to wrap around (mod $\pi$), so $[a, b]$ indicates $[a,\pi)\cup[0,b]$ if $a > b$.
We assume that no polygons contain parallel adjacent edges, as any such edges can be merged.
%Additions/subtractions of angles and angle ranges will typically be done modulo $\pi$ (i.e., with ``wrap-around'').

We assume that all polygon boundaries and their edges/arcs are oriented in ccw order.
For each edge $e$ of $Q$, let $\ext{e}$ denote its extension as an oriented line (with $Q$ on its ``left'' side).  For an arc $\Gamma$, let $\ext{\Gamma}$ denote an extension of $\Gamma$ where the first and last edge are extended to rays (again oriented with $Q$ on 
its ``left'' side).
We use $\OO$ notation to hide $\polylog n$ factors.

\section{3-Contact (1-Anchor) Case}\label{sec:3contact}

In this section, we solve the 3-contact case, where there is 1 anchor vertex $p_1$ of $P$.  We will first present an algorithm for $k=3$ (when $P$ is a triangle), and then we discuss how to generalize it for $k>3$. 
We will divide into sub-problems operating on different arcs.
For $k=3$, the goal is to find a placement where the anchor $p_1$ is at a vertex on an arc $\Gamma_1$, and
the two other vertices $p_2$ and $p_3$ of $P$ are on edges of arcs $\Gamma_2$ and~$\Gamma_3$.

\subsection{An Easy ``Disjoint'' Case for $k=3$}

We begin with an easy lemma to handle the case when the angle ranges for $\Gamma_2$ and $\Gamma_3$, after suitable rotational shifts, are disjoint.  
%(We can ignore (ii) below for now, but it will be useful later when generalizing to $k>3$.) %\todo{add a diagram for this}

\begin{figure}
    \centering    \includegraphics[width=.45\textwidth]{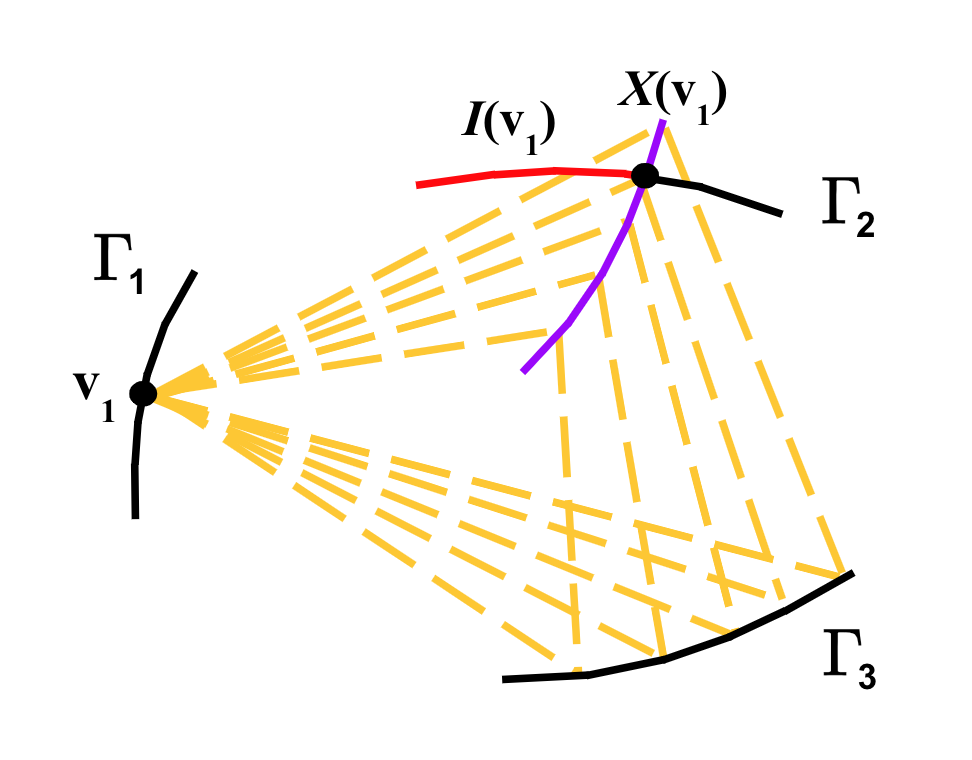}
    \caption{The construction from Lemma~\ref{lem:int}. Each dashed triangle is a similar copy of $\triangle p_1p_2p_3$, and the purple arc is $f_{v_1}(\Gamma_3)$. (We draw $\Gamma_2$ and $\Gamma_3$ instead of $\ext{\Gamma_2}$ and $\ext{\Gamma_3}$ for visual clarity.)}
    \label{fig:int}
\end{figure}

\begin{lemma}\label{lem:int}
Let $\triangle p_1p_2p_3$ be a triangle.
Let $\Gamma_1,\Gamma_2,\Gamma_3$ be arcs of a convex $n$-gon $Q$, 
\LIPICS{s.t.}%
\NORMAL{such that}

\begin{itemize}
\item $\Lambda(\Gamma_2)+\theta_{p_1p_3}$ and $\Lambda(\Gamma_3)+\theta_{p_1p_2}$ are disjoint (mod $\pi$).
\end{itemize}

\noindent
In $\OO(|\Gamma_1|)$ time, we can compute, for each vertex $v_1$ of $\Gamma_1$, a point $X(v_1)$ on $\Gamma_2$ and a sub-arc $\INT(v_1)$ of $\Gamma_2$, satisfying the following property:\footnote{We allow $X(v_1)$ to be undefined and $\INT(v_1)$ to be empty.}
\begin{quote}
For every similarity transformation $\varphi$ that has $\varphi(p_1)$ being a vertex $v_1$ of $\Gamma_1$ and $\varphi(p_2)$ on $\Gamma_2$, we have: (i)~$\varphi(p_3)$ is on $\Gamma_3$ iff $\varphi(p_2)=X(v_1)$, and
(ii)~$\varphi(p_3)$ is left of $\ext{\Gamma_3}$ iff $\varphi(p_2)$ is on $\INT(v_1)$.
\end{quote}
\end{lemma}

\begin{proof}
Let $f_{v_1}(\zeta)$ be the point $\varphi(p_2)$ for the unique
similarity transformation $\varphi$ with $\varphi(p_1)=v_1$ and $\varphi(p_3)=\zeta$.
In other words, $f_{v_1}$ is a similarity transformation that keeps $v_1$ fixed and sends $p_3$ to $p_2$, i.e., we rotate around $v_1$ by an angle $\theta_{p_1p_2}-\theta_{p_1p_3}+\{0,\pm\pi\}$, and scale by factor $\|p_1-p_2\|/\|p_1-p_3\|$.
Thus, $f_{v_1}(\Gamma_3)$ is a similar copy of $\Gamma_3$.
%Then $f_{v_1}(\Gamma_3)$ is a similar copy of $\Gamma_3$ (rotated around $v_1$ by an angle $\theta_{p_1p_2}-\theta_{p_1p_3}+\{0,\pm\pi\}$, and scaled by factor $\|p_1-p_2\|/\|p_1-p_3\|$; this transformation maps $p_3$ to $p_2$).  
The supporting lines for $f_{v_1}(\Gamma_3)$ have angles in $\Lambda(\Gamma_3)+\theta_{p_1p_2}-\theta_{p_1p_3}+\{0,\pm\pi\}$, which by the assumption is disjoint from $\Lambda(\Gamma_2)$ (mod $\pi$).
Thus,\footnote{
This is analogous to the fact that if
$f$ and $g$ are functions over $\R$ where the ranges of their
derivatives $f'$ and $g'$ are contained in two disjoint closed intervals, then $f$ and $g$ intersect once.
}
$f_{v_1}(\ext{\Gamma_3})$ and $\ext{\Gamma_2}$ intersect once, at a unique point~$\nu$, which we define as $X(v_1)$, and which can be computed by binary search (see Figure~\ref{fig:int}). 
We can define $\INT(v_1)$ to be a prefix or suffix of $\Gamma_2$ delimited by $X(v_1)$.
\end{proof}

Thus, to solve the 3-contact problem for $k=3$, we can just examine the unique similarity transformation $\varphi$ with $\varphi(p_1)=v_1$ and $\varphi(p_2)=X(v_1)$ for each vertex $v_1$ of $\Gamma_1$, in near-linear total time (assuming the disjointness condition is met). It suffices to consider only this single similarity transformation for each $v_1 \in \Gamma_1$, since this is the only possible 3-contact placement.

\IGNORE{
\begin{lemma}\label{lem:tri:rs}
Let $\triangle p_1p_2p_3$ be a triangle.
Let $\Gamma_1,\Gamma_2,\Gamma_3$ be arcs of a convex $n$-gon $Q$, such that

\begin{enumerate}
\item $\Lambda(\Gamma_2)+\theta_{p_1p_3}$ and $\Lambda(\Gamma_3)+\theta_{p_1p_2}$ are disjoint, modulo $\pi$.
\end{enumerate}

In $\OO(n)$ time, we can find a similarity transformation $\varphi$, maximizing the scaling factor,
such that $\varphi(p_1)$ is a vertex of $\Gamma_1$, and $\varphi(p_2)$ is on $\Gamma_2$, and $\varphi(p_3)$ is on $\Gamma_3$.
\end{lemma}
\begin{proof}
For each vertex $v_1$ of $\Gamma_1$, we just check the similarity transformation $\varphi$ with $\varphi(p_1)=v_1$ and
$\varphi(p_2)=X(v_1)$.
\end{proof}
}

\subsection{A ``Double-Disjoint'' Case for $k=3$}

Next, we address a different case where the angle ranges for $\Gamma_1$ and $\Gamma_2$ are disjoint and the angle
ranges for $\Gamma_1$ and $\Gamma_3$ are disjoint, after appropriate rotational shifts.  The following lemma reveals a crucial monotonicity phenomenon that we will repeatedly exploit.

To state the lemma, we first introduce some definitions:
For two arcs $\Gamma_1$ and $\Gamma_2$, a \emph{pairing} $\MATCH$ between $\Gamma_1$ and $\Gamma_2$ refers to a subdivision of the (straight) edges of
$\Gamma_1$ and $\Gamma_2$ into sub-edges, together with a bijective mapping between the sub-edges of $\Gamma_1$ and the sub-edges of $\Gamma_2$.  For a sub-edge $e_1$ of $\Gamma_1$, we use $\MATCH(e_1)$ to denote $e_1$'s corresponding sub-edge in $\Gamma_2$; similarly, for
a sub-edge $e_2$ of $\Gamma_2$, we use $\MATCH(e_2)$ to denote $e_2$'s corresponding sub-edge in $\Gamma_1$.
We say that the pairing $\MATCH$ is \emph{monotonically increasing (resp.\ decreasing)} if $\MATCH(e_1)$ always advances in ccw (resp.\ cw) order in $\Gamma_2$ as $e_1$ advances in ccw order in $\Gamma_1$ (see Figure~\ref{fig:mon}). 

\begin{figure}
\centering
\includegraphics[width=.55\textwidth]{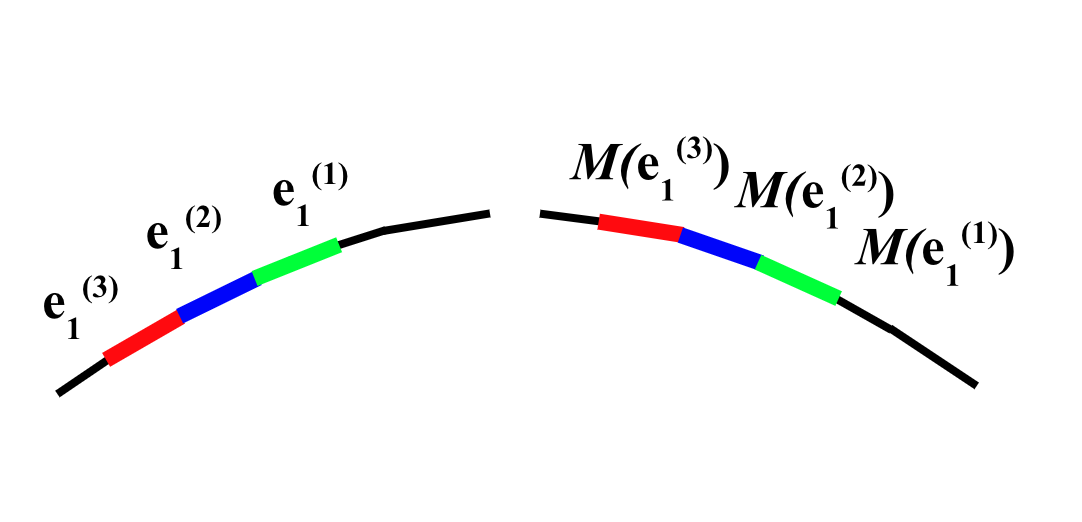}
\hspace{.05\textwidth}
\includegraphics[width=.34\textwidth]{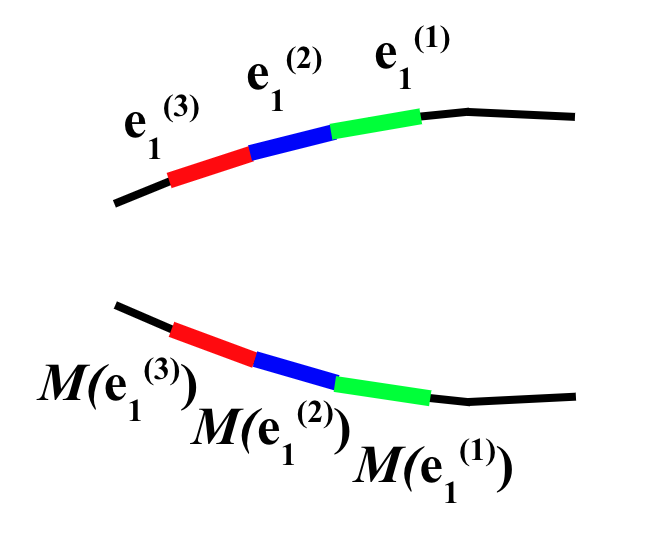}
%\begin{subfigure}{.62\textwidth}
%  \centering
%  \includegraphics[width=.9\textwidth]{Figures/MonInc.png}
%  \label{fig:sub2}
%\end{subfigure}%
%\begin{subfigure}{.38\textwidth}
%  \centering
%  \includegraphics[width=.9\textwidth]{Figures/MonDec.png}
%  \label{fig:sub1}
%\end{subfigure}
\caption{(Left) Monotonically increasing pairing. (Right) Monotonically decreasing pairing.}
\label{fig:mon}
\end{figure}

%\todo{add a diagram}
\begin{lemma}\label{lem:tri:match} \emph{(Pairing Lemma)}\ \
Let $\triangle p_1p_2p_3$ be a triangle.
Let $\Gamma_1,\Gamma_2,\Gamma_3$ be arcs of a convex $n$-gon $Q$, such that

\begin{enumerate}
\item $\Lambda(\Gamma_1)+\theta_{p_2p_3}$ and $\Lambda(\Gamma_2)+\theta_{p_1p_3}$ are disjoint (mod $\pi$), and
\item $\Lambda(\Gamma_1)+\theta_{p_2p_3}$ and $\Lambda(\Gamma_3)+\theta_{p_1p_2}$ are disjoint (mod $\pi$).
\end{enumerate}

\noindent
In $\OO(n)$ time, we can compute a (monotonically increasing or decreasing) pairing $\MATCH$ between $\Gamma_2$ and $\Gamma_3$ with $O(n)$ sub-edges, satisfying the following property:
\begin{quote}
For every similarity transformation $\varphi$ that has $\varphi(p_1)$ on $\Gamma_1$ and $\varphi(p_2)$ on a sub-edge $e_2$ of $\Gamma_2$, we have: (i)~$\varphi(p_3)$ is on $\Gamma_3$ iff $\varphi(p_3)$ is on $\MATCH(e_2)$;
and (ii)~$\varphi(p_3)$ is left of $\ext{\Gamma_3}$ iff $\varphi(p_3)$ is left of $\ext{\MATCH(e_2)}$.
\end{quote}
\end{lemma}
\begin{proof}
We \emph{match} a point $\mu$ on $\ext{\Gamma_2}$ with a point $\nu$ on $\ext{\Gamma_3}$ iff
there exists a similarity transformation $\varphi$ with $\varphi(p_2)=\mu$, $\varphi(p_1)$ on $\ext{\Gamma_1}$, and $\varphi(p_3)=\nu$.

Observe that a point $\mu$ on $\ext{\Gamma_2}$ matches a unique point $\nu$ on $\ext{\Gamma_3}$.
To see this, let $f_\mu(\zeta)$ be the point $\varphi(p_3)$ for the unique
similarity transformation $\varphi$ with $\varphi(p_2)=\mu$ and $\varphi(p_1)=\zeta$.
In other words, $f_\mu$ is the similarity transformation that keeps $\mu$ fixed and sends $p_1$ to $p_3$, i.e., we rotate around $\mu$ by an angle $\theta_{p_2p_3}-\theta_{p_2p_1}+\{0,\pm\pi\}$, and scale by factor $\|p_3-p_2\|/\|p_1-p_2\|$.
Thus, $f_\mu(\Gamma_1)$ is a similar copy of $\Gamma_1$.
%Then $f_\mu(\Gamma_1)$ is a similar copy of $\Gamma_1$ (rotated around $\mu$ by an angle $\theta_{p_2p_3}-\theta_{p_2p_1}+\{0,\pm\pi\}$, and scaled by factor $\|p_3-p_2\|/\|p_1-p_2\|$; this transformation maps $p_1$ to $p_3$).  
The supporting lines for $f_\mu(\Gamma_1)$ have angles in $\Lambda(\Gamma_1)+\theta_{p_2p_3}-\theta_{p_2p_1}+\{0,\pm\pi\}$, which by assumption~2 is disjoint from $\Lambda(\Gamma_3)$ (mod $\pi$).
Thus, $f_\mu(\ext{\Gamma_1})$ and $\ext{\Gamma_3}$ intersect once, namely, at the unique point $\nu$. 
A symmetric argument (swapping subscripts 2 and 3) shows that a point $\nu$ on $\ext{\Gamma_3}$ matches a unique point $\mu$ on $\ext{\Gamma_2}$,
this time, by assumption~1.

Consequently,\footnote{This is analogous to the fact 
 that a continuous bijective function over $\R$ must be monotone.
} as $\mu$ moves along $\ext{\Gamma_2}$, its matching point $\nu$ moves monotonically along $\ext{\Gamma_3}$.
We break an edge at the points $\mu$ on $\Gamma_2$ that match the vertices of $\ext{\Gamma_3}$, which can be found by $n$ binary searches.
Similarly, we break an edge at the points $\nu$ on $\Gamma_3$ that match the vertices of $\ext{\Gamma_3}$, which can be found by $n$ binary searches.
As a result, all points $\mu$ on a common sub-edge $e_2$ of $\ext{\Gamma_2}$ are matched with points on a common sub-edge 
of  $\ext{\Gamma_3}$, which we define as $\MATCH(e_2)$.  For all these points $\mu$, we have $f_\mu(\ext{\Gamma_1})$ intersecting this sub-edge $\MATCH(e_2)$ of $\ext{\Gamma_3}$.
Also, for all $\zeta\in \Gamma_1$, $f_\mu(\zeta)$ is left of $\ext{\Gamma_3}$ iff $f_\mu(\zeta)$ is left of $\ext{\MATCH(e_2)}$ (see Figure~\ref{fig:tri:match}).
\end{proof}

\begin{figure}
    \centering    \includegraphics[width=.45\textwidth]{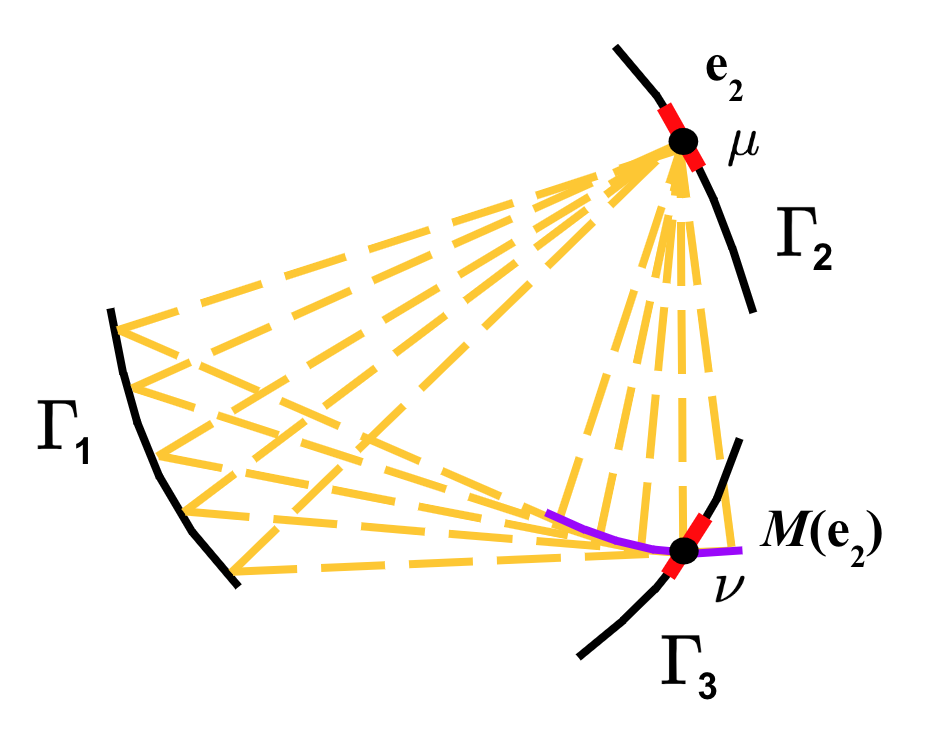}
    \caption{An example of a pairing. Each dashed triangle is a similar copy of $\triangle p_1p_2p_3$.}
    \label{fig:tri:match}
\end{figure}

In the above, we did not claim a monotone pairing between $\Gamma_2$ and $\Gamma_1$, nor between $\Gamma_1$ and $\Gamma_3$. Otherwise, we would get a linear upper bound on the number of 3-contact solutions in this case, which by our subsequent divide-and-conquer algorithm would yield an $O(n\log n)$ bound on the number of 3-contact solutions in general for $k=3$,
contradicting the known quadratic lower bound~\cite{AgarwalAS, LeeEomAhn}! This contradiction does not arise since in the worst case, each matched pair from $\Gamma_2$ and $\Gamma_3$ could admit legal 3-contact placements with \emph{every} vertex of $\Gamma_1$.

With the Pairing Lemma at hand, we can efficiently solve the problem when the two disjointness conditions are met. Specifically, we set up a range searching sub-problem between the $O(n)$ pairs of sub-edges in $\Gamma_2$ and $\Gamma_3$ (the ``data set''), and 
the $O(n)$ vertices of $\Gamma_1$ (the ``query points''). This range searching sub-problem turns out to be near-linear-time solvable:

\begin{lemma}\label{lem:tri:rs}
Let $\triangle p_1p_2p_3$ be a triangle.
Let $\Gamma_1,\Gamma_2,\Gamma_3$ be arcs of a convex $n$-gon $Q$, 
\LIPICS{s.t.}%
\NORMAL{such that}

\begin{enumerate}
\item $\Lambda(\Gamma_1)+\theta_{p_2p_3}$ and $\Lambda(\Gamma_2)+\theta_{p_1p_3}$ are disjoint (mod $\pi$), and
\item $\Lambda(\Gamma_1)+\theta_{p_2p_3}$ and $\Lambda(\Gamma_3)+\theta_{p_1p_2}$ are disjoint (mod $\pi$).
\end{enumerate}

\noindent
In $\OO(n)$ time, we can find a similarity transformation $\varphi$, maximizing the scaling factor,
such that $\varphi(p_1)$ is a vertex of $\Gamma_1$, $\varphi(p_2)$ is on $\Gamma_2$, and $\varphi(p_3)$ is on $\Gamma_3$.
\end{lemma}
\begin{proof}
Given $(s,t,u,v)\in\R^4$, let $\varphi_{s,t,u,v}:\R^2\rightarrow\R^2$ be the similarity transformation $(x,y)\mapsto 
(sx-ty+u,tx+sy+v)$ (which has scaling factor $\sqrt{s^2+t^2}$).  %Let $\sigma=\{(s,t,u,v): s^2+t^2=1\}$.

Apply Lemma~\ref{lem:tri:match} to get a pairing $M$ between $\Gamma_2$ and $\Gamma_3$.
For each sub-edge $e_2$ of $\Gamma_2$,
define 
\[ \begin{array}{rrl}
\PPP(e_2)\ = &  \{(s,t,u,v)\in\R^4:  & \mbox{$\varphi_{s,t,u,v}(p_2)$ is on $e_2$, and 
$\varphi_{s,t,u,v}(p_3)$ is on $\MATCH(e_2)$}\}\\[1ex]
R_\rho(e_2)\ = & \{\varphi_{s,t,u,v}(p_1): & \mbox{$(s,t,u,v)\in\PPP(e_2)$ and $s^2+t^2\ge\rho^2$}\}.
\end{array}
\]
Observe that $\PPP(e_2)$ is a 2-dimensional convex polygon in $\R^4$ with $O(1)$ edges (since 
$\varphi_{s,t,u,v}(p_2)$ and $\varphi_{s,t,u,v}(p_3)$ are linear in the variables $s,t,u,v$, and the 2 point-on-line-segment conditions yield
2 linear equality constraints and
4 linear inequality constraints in these 4 variables).
%Define 
%\[ R_\rho(e_2)\ =\ \{\varphi_{s,t,u,v}(p_1): \mbox{$(s,t,u,v)\in\PPP(e_2)$ and $s^2+t^2\ge\rho^2$}\}.
%\]
Furthermore, $R_\rho(e_2)$ is a region in $\R^2$ which is the intersection of a convex $O(1)$-gon with the
exterior of an ellipse (since $(s,t,u,v)\mapsto \varphi_{s,t,u,v}(p_1)$ is a linear projection from $\R^4$ to $\R^2$, and 
the projection of a 2-dimensional slice of the cylinder $\{(s,t,u,v): s^2+t^2=\rho^2\}$ gives an ellipse).

The decision problem (deciding whether the maximum scaling factor is at least a given value $\rho$) reduces to finding a pair of vertex $v_1$ of $\Gamma_1$ and sub-edge $e_2$ of $\Gamma_2$, such
that $v_1\in R_\rho(e_2)$.
To this end, we will build a data structure to store the $O(n)$ regions $R_\rho(e_2)$ over all $e_2$ so that we can quickly decide whether the query point $v_1$ stabs (i.e., is contained in) some
 region $R_\rho(e_2)$.  

We use standard techniques in geometric data structures.
First, we consider the \emph{range stabbing} problem for exteriors of ellipses: build a data structure for a set of $O(n)$ ellipses in $\R^2$, so that we can quickly decide whether a query point
stabs the exterior of some ellipse, i.e., whether a query point is outside the intersection of the interiors of the ellipses.
The intersection of the interiors of $O(n)$ ellipses (which is a single cell in the arrangement) has almost linear combinatorial complexity by standard results on Davenport-Schinzel sequences~\cite{SharirBOOK},
and can be constructed in $\OO(n)$ time, e.g., by divide-and-conquer.
Thus, this problem can be solved with $\OO(n)$ preprocessing time and $\OO(1)$ query time.

Next, we consider range stabbing for our regions $R_\rho(e_2)$.  As each region is the intersection of
a convex $O(1)$-gon with the exterior of an ellipse, we can use standard \emph{multi-level}
data structuring techniques%
\LIPICS{%
~\cite{AgarwalEricksonSURVEY}%
}%
\NORMAL{%
\footnote{
Multi-level data structures are for solving range-searching-related problems, where in a query, we seek an object satisfying a conjunction of multiple ($O(1)$) constraints.  The idea is to take a tree structure for solving the problem with one constraint, and have each node of the tree recursively store structures for solving the problem with the remaining constraints.  See a general survey on range searching
~\cite{AgarwalEricksonSURVEY} for more details.
}%
}
to handle the extra $O(1)$ halfplane constraints.  Generally, halfplane range searching cannot
be solved with near-linear preprocessing time and polylogarithmic query time.  But in our application, all query points
$v_1$ lie on a convex chain $\Gamma_1$.  The constraint that such a query point $v_1$ lies inside a halfplane
is equivalent to the condition that $v_1$ lies inside one of $O(1)$ 1D intervals, assuming that the vertices of $\Gamma_1$ are stored in a sorted array. %\cite{BergCKO08}. 
We can therefore use 1D range
trees~\cite{AgarwalEricksonSURVEY,Lee04,PreparataS} to handle the halfplane constraints, with only a logarithmic factor increase in the preprocessing and query time.

The optimization problem reduces to the decision problem by a standard application of  parametric search~\cite{Megiddo}.
(The application requires a parallelization of the decision algorithm: the preprocessing part, namely, the construction of the intersection of interiors of ellipses, is straightforwardly parallelizable by divide-and-conquer; the $O(n)$ queries can trivially be answered in parallel.)  Parametric search increases the running time by a polylogarithmic factor.
Alternatively, we can apply Chan's randomized optimization technique~\cite{Chan99}, which avoids extra factors. (The application here is straightforward, since the problem can be viewed as a generalized ``closest-pair-type'' problem~\cite{Chan99} between two sets of objects.)  
\end{proof}

\subsection{Simple Divide-and-Conquer Algorithm for $k=3$}

We now have all the ingredients needed to put together a simple recursive algorithm to solve
the $k=3$ problem in the 3-contact case:

\begin{theorem}\label{thm:tri:3contact}
Given a triangle $P$ and a convex $n$-gon $Q$, in $\OO(n)$ time,
we can find the largest similar copy of $P$ contained in $Q$ that has $1$ vertex of $P$ at a vertex of $Q$ and
the $2$ other vertices of $P$ on edges of $Q$.
\end{theorem}
\begin{proof}
Let $P=\triangle p_1p_2p_3$.
Arbitrarily divide $\partial Q$ into $O(1)$ arcs, and let $\Gamma_1,\Gamma_2,\Gamma_3$ be 3 such arcs (allowing duplicates).  We will try all $O(1)$ choices of $\Gamma_1,\Gamma_2,\Gamma_3$.

Let $S$ be an interval.
Let $\Gamma_1(S)$ (resp.\ $\Gamma_2(S)$ and $\Gamma_3(S)$) be the sub-arc of $\Gamma_1$
(resp.\ $\Gamma_2$ and $\Gamma_3$) consisting of all edges whose supporting lines have angles
in $S-\theta_{p_2p_3}$  (resp.\ $S-\theta_{p_1p_3}$ and $S-\theta_{p_1p_2}$) (mod $\pi$).
We will recursively solve the following problem: find a similarity transformation $\varphi$, maximizing the scaling factor, such that 
$\varphi(p_1)$ is a vertex $v_1$ of $\Gamma_1(S)$, $\varphi(p_2)$ is on $\Gamma_2(S)$, and $\varphi(p_3)$ is on $\Gamma_3(S)$.

As a first step, we remove edges not participating in $\Gamma_1(S),\Gamma_2(S),\Gamma_3(S)$, so that the number of edges in $Q$ is reduced to $m(S):=|\Gamma_1(S)|+|\Gamma_2(S)|+|\Gamma_3(S)|$.
Partition $S$ into sub-intervals $S^-$ and $S^+$ (as shown in Figure~\ref{fig:tri:3contact}) so that $m(S^-),m(S^+) = m(S)/2 \pm O(1)$. 
%(Such a partition is always possible because, for any length zero interval $\theta$, $m(\theta) \le O(1)$.)
We try various possibilities and take the best solution found:

\begin{itemize}
\item Case 1: $\varphi(p_1)$ is on $\Gamma_1(S^-)$, $\varphi(p_2)$ is on $\Gamma_2(S^-)$, and $\varphi(p_3)$ is on $\Gamma_3(S^-)$.
We can recursively solve the problem for $S^-$.
\item Case 2: $\varphi(p_1)$ is on $\Gamma_1(S^+)$, $\varphi(p_2)$ is on $\Gamma_2(S^+)$, and $\varphi(p_3)$ is on $\Gamma_3(S^+)$.
We can recursively solve the problem for $S^+$.
\item Case 3: $\varphi(p_1)$ is on $\Gamma_1(S^-)$, $\varphi(p_2)$ is on $\Gamma_2(S^+)$, and $\varphi(p_3)$ is on $\Gamma_3(S^+)$.
Since $\Lambda(\Gamma_1(S^-))+\theta_{p_2p_3}\subseteq S^-$ and $\Lambda(\Gamma_2(S^+))+\theta_{p_1p_3}\subseteq S^+$ are disjoint (mod $\pi$),
and $\Lambda(\Gamma_1(S^-))+\theta_{p_2p_3}\subseteq S^-$ and $\Lambda(\Gamma_3(S^+))+\theta_{p_1p_2}\subseteq S^+$ are disjoint (mod $\pi$),
we can solve this sub-problem by Lemmas \ref{lem:tri:match}--\ref{lem:tri:rs} in $\OO(m(S))$ time.
\item Case 4: $\varphi(p_1)$ is on $\Gamma_1(S^-)$, $\varphi(p_2)$ is on $\Gamma_2(S^-)$, and $\varphi(p_3)$ is on $\Gamma_3(S^+)$.
Since $\Lambda(\Gamma_2(S^-))+\theta_{p_1p_3}\subseteq S^-$ and $\Lambda(\Gamma_3(S^+))+\theta_{p_1p_2}\subseteq S^+$ are disjoint (mod $\pi$),
we can solve this sub-problem by Lemma~\ref{lem:int} in $\OO(m(S))$ time. Namely, for each vertex $v_1$ of $\Gamma_1(S^-)$, we just check 
the unique similarity transformation $\varphi$ with $\varphi(p_1)=v_1$ and $\varphi(p_2)=X(v_1)$.
%feasibility of the similarity transformation $\varphi$ with $\varphi(p_1)=v_1$ and $\varphi(p_2)=X(v_1)$, which is fast since we can tell whether any given point is inside $Q$ via binary search \cite{PreparataS}.
%\cite{BergCKO08}.
\end{itemize}

\noindent
All remaining cases are symmetric to Cases 3 and 4
(swapping subscripts 2 and 3 and/or $S^-$ and $S^+$).

Letting $m=m(S)$, we obtain the following recurrence for the running time:
\[ T(m)\: \le\: 2\,T(m/2 + O(1)) + \OO(m).
\]
The recurrence solves to $T(m)=\OO(m)$.
\end{proof}

It is not difficult to modify the algorithm to also solve the 2-contact (2-anchor) case, as shown in 
\LIPICS{the full paper.}%
\NORMAL{Appendix~\ref{app:2contact}.}
This gives a complete $\OO(n)$ time algorithm for $k=3$.

\begin{figure}
    \centering    \includegraphics[width=.48\textwidth]{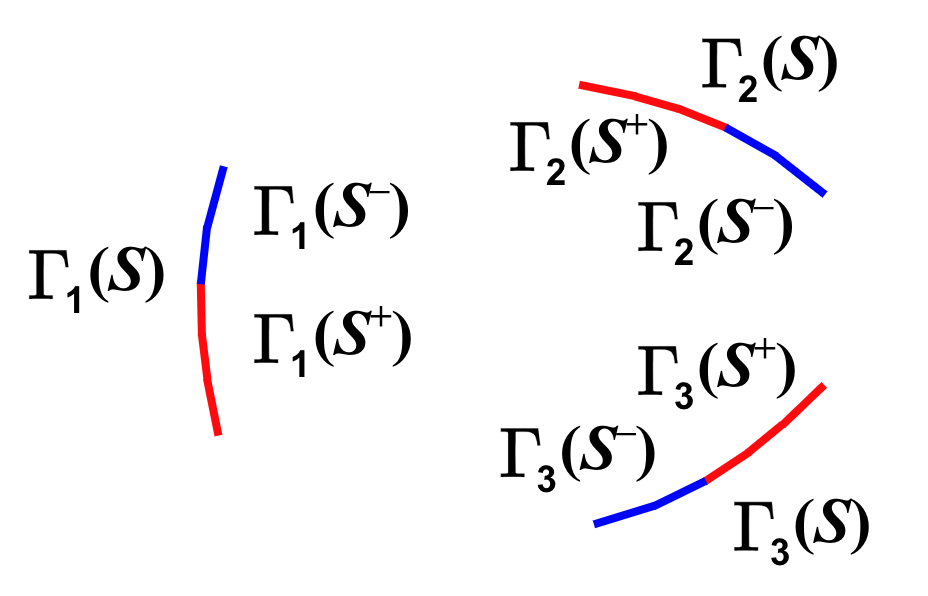}
    \caption{Partitioning each $\Gamma_i(S)$ into $\Gamma_i(S^+)$ and $\Gamma_i(S^-)$, for $i \in \{1,2,3\}$.}
    \label{fig:tri:3contact}
\end{figure}
 
\subsection{Generalizing to $k>3$}

With further effort, we can also solve the problem for general $k$ in the 3-contact case.
Say $p_1$ is the anchor vertex, and $p_2$ and $p_3$ are the other two vertices of $P$ in contact with $Q$.
The idea is to just apply the Pairing Lemma to the triangles $\triangle p_1p_2p_i$ for all other vertices $p_i$ of $P$,
assuming appropriate disjointness conditions.
We need to extend the problem to equip each vertex $v_1$ of $\Gamma_1$ with a sub-arc $\INT(v_1)$ of $\Gamma_2$ which restricts
the placement of~$p_2$.
The resulting range searching sub-problems can be solved in a manner similar to the triangle case, which we show in the following extension of Lemma~\ref{lem:tri:rs}%
\LIPICS{ (see the full paper for the proof)}:

\begin{lemma}\label{lem:kgon:rs}
Let $P$ be a $k$-gon with vertices $p_1,\ldots,p_k$ (not necessarily in sorted order).
Let $\Gamma_1,\ldots,\Gamma_k$ be arcs of a convex $n$-gon $Q$, such that

\begin{enumerate}
\item for each $i\in\{3,\ldots,k\}$, $\Lambda(\Gamma_1)+\theta_{p_2p_i}$ and
$\Lambda(\Gamma_2)+\theta_{p_1p_i}$ are disjoint (mod $\pi$), and
\item for each $i\in\{3,\ldots,k\}$, $\Lambda(\Gamma_1)+\theta_{p_2p_i}$ and
$\Lambda(\Gamma_i)+\theta_{p_1p_2}$ are disjoint (mod $\pi$).
\end{enumerate}

\noindent
For each vertex $v_1$ of $\Gamma_1$, we are given a sub-arc $\INT(v_1)$ of $\Gamma_2$. 
In $\OO(k^2n)$ time, we can find a similarity transformation $\varphi$, maximizing the scaling factor,
such that $\varphi(p_1)$ is a vertex $v_1$ of $\Gamma_1$, $\varphi(p_2)$ is on $\INT(v_1)\subseteq \Gamma_2$, $\varphi(p_3)$ is on $\Gamma_3$, and for each $i\in\{4,\ldots,k\}$, $\varphi(p_i)$ is left of $\ext{\Gamma_i}$.
\end{lemma}

\NORMAL{
\begin{proof}
We generalize the proof of Lemma~\ref{lem:tri:rs}.
%As before, given $(s,t,u,v)\in\R^4$, let $\varphi_{s,t,u,v}:\R^2\rightarrow\R^2$ be the similarity transformation $(x,y)\mapsto 
%(sx-ty+u,tx+sy+v)$ (which has scaling factor $\sqrt{s^2+t^2}$).  %Let $\sigma=\{(s,t,u,v): s^2+t^2=1\}$.
For each $i\in\{3,\ldots,k\}$, we apply Lemma~\ref{lem:int} to $\triangle p_1p_2p_i,\Gamma_1,\Gamma_2,\Gamma_i$, to
obtain a pairing $\MATCH_{i}$ between $\Gamma_2$ and $\Gamma_i$.
We then overlay the $O(k)$ subdivisions of $\Gamma_2$. In doing so, the mapping $M_{i}$ from sub-edges of $\Gamma_2$ to sub-edges of $\Gamma_i$
may not be bijective, but this is fine. The number of sub-edges is $n'=O(kn)$. 
We further subdivide $\Gamma_2$ at the endpoints of the sub-arcs $\INT(v_1)$ for every vertex $v_1$ of $\Gamma_1$. As a consequence, the following property holds:
\begin{quote}
For every similarity transformation $\varphi$ that has $\varphi(p_1)$ being a vertex $v_1$ of $\Gamma_1$ and $\varphi(p_2)$ on a sub-edge $e_2$ of $\Gamma_2$, we have: (i)~$\varphi(p_3)$ is on $\Gamma_3$ iff $\varphi(p_3)$ is on $\MATCH_{3}(e_2)$, and (ii)~for each $i\in\{4,\ldots,k\}$, $\varphi(p_i)$ is left of $\ext{\Gamma_i}$ iff $\varphi(p_i)$ is left of $\ext{\MATCH_{i}(e_2)}$.
\end{quote}

For each sub-edge $e_2$ of $\Gamma_2$,
define 
\[ \begin{array}{rrl}
\PPP(e_2)\ = & \{(s,t,u,v)\in\R^4:  & \mbox{$\varphi_{s,t,u,v}(p_2)$ is on $e_2$, and $\varphi_{s,t,u,v}(p_3)$ is on $\MATCH_{3}(e_2)$, and} \\[.5ex]
 &&\mbox{$\varphi_{s,t,u,v}(p_i)$ is left of $\ext{\MATCH_{i}(e_2)}$ for all $i\in\{4,\ldots,k\}$}\}\\[1ex]
 R_\rho(e_2)\ = & \{\varphi_{s,t,u,v}(p_1): & \mbox{$(s,t,u,v)\in\PPP(e_2)$ and $s^2+t^2\ge\rho^2$}\}.
\end{array}
\]
Similar to before, $\PPP(e_2)$ is a 2-dimensional polygon in $\R^4$, with $O(k)$ edges, and $R_\rho(e_2)$ is a region in $\R^2$ which is the intersection of a convex $O(k)$-gon with the
exterior of an ellipse. 

The decision problem reduces to finding a pair of vertex $v_1$ of $\Gamma_1$ and sub-edge $e_2$ of $\Gamma_2$, such
that $v_1\in R_\rho(e_2)$ and $e_2$ is in $\INT(v_1)$.
To this end, we will build a data structure to store the $O(n')$ regions $R_\rho(e_2)$ over all $e_2$ so that, given a query point $v_1$ and a query sub-arc~$I$, we can quickly decide whether the query point $v_1$ stabs some
 region $R_\rho(e_2)$ with $e_2$ in the query sub-arc~$I$.  

Each such region can be decomposed into $O(k)$ sub-regions, where
each sub-region is the intersection of the exterior of an ellipse with $O(1)$ halfplanes (the preprocessing time is increased by an $O(k)$ factor).  As before, we can use ellipse range stabbing combined with multi-level data structuring techniques (1D range tree).
The constraint that $e_2$ is in the query sub-arc $I$ (which can be represented as a 1D interval) can again be handled
by multi-leveling, with another level of 1D range trees.  The preprocessing time is $\OO(kn')$ and the query time
is $\OO(1)$.  The optimization problem reduces to the decision problem by standard
parametric search (or randomized search) as before.  
\end{proof}
}

We now give a slightly more intricate divide-and-conquer algorithm for general $k$:

\begin{theorem}\label{thm:kgon:3contact}
Given a $k$-gon $P$ and a convex $n$-gon $Q$ (where $k \le n$),
we can find the largest similar copy of $P$ contained in $Q$ that has $1$ vertex of $P$ at a vertex of $Q$ and
the $2$ other vertices of $P$ on $2$ edges of $Q$,
in $O(k^{O(1/\eps)}n^{1+\eps})$ time for any $\eps>0$.  % $n2^{O(\sqrt{\log k\log n})}$ time.
\end{theorem}
\begin{proof}
Suppose the vertices of $P$ are $p_1,\ldots,p_k$ (not necessarily in sorted order).
Divide $\partial Q$ into $O(1)$ arcs, and let $\Gamma_1,\Gamma_2,\Gamma_3$ be 3 such arcs (allowing duplicates).
We will try all choices for $p_1,p_2,p_3$ and $\Gamma_1,\Gamma_2,\Gamma_3$; this increases the final running time by a factor of $O(k^3)$.

Let $\Gamma_4, \ldots, \Gamma_{k}$ be arcs of $\partial Q$, so that a similarity transformation $\varphi$ has $\varphi(P)$ inside $Q$ iff $\varphi(p_i)$ is left of $\ext{\Gamma_i}$ for all $i \in \{1,\dots k\}$. This is w.l.o.g.\ since we can just make $O(1)$ copies of $p_4,\ldots,p_k$ and associate each copy with an arc of $\partial Q$, while increasing $k$ by a constant factor. Note that duplicate arcs are allowed, and  some of these arcs may even be the same as $\Gamma_1$, $\Gamma_2$, or $\Gamma_3$.

We will describe a recursive algorithm, where the input consists of $k$ arcs
$\langle\Gamma_1,\ldots,\Gamma_k\rangle$ together with
a sub-arc $\INT(v_1)\subseteq\Gamma_2$ for every vertex $v_1$ of $\Gamma_1$.  (For the initial problem, $\INT(v_1)$ will be all of $\Gamma_2$ for all $v_1$.)  Our algorithm will find a similarity transformation $\varphi$, maximizing the scaling factor, such that $\varphi(p_1)$ is a vertex $v_1$ of $\Gamma_1$, $\varphi(p_2)$ is on $\INT(v_1)$, $\varphi(p_3)$ is on $\Gamma_3$, and for all $i\in\{4,\ldots,k\}$, $\varphi(p_i)$ is left of 
$\ext{\Gamma_i}$.

%Let the vertices of $\Gamma_1$ be $v_1^{(1)}, \dots, v_1^{(|\Gamma_1|)}$. Our goal is to find the largest placement of $P$ within $Q$ such that $\varphi(p_1)$ is a vertex $v_1$ of $\Gamma_1$, $\varphi(p_2)$ is on $\Gamma_2$, and $\varphi(p_3)$ is on $\Gamma_3$. To this end, we will recursively produce and solve sub-problems of the form $(\Gamma_1, \dots \Gamma_{k}, \INT(v_1^{(1)}), \ldots, \INT(v_1^{(|\Gamma_1|)}))$ that have the following property:
%\begin{quote}
%    Let $\varphi$ be a similarity transformation and $v_1$ a vertex of $\Gamma_1$. Say $\varphi(p_1) = v_1$ and $\varphi(p_2)$ is on $\INT(v_1)\subseteq\Gamma_2$. Then $\varphi(P)$ is inside $Q$ iff  for all $i \in \{3,\dots, k\}$, $\varphi(p_i)$ is left of $\ext{\Gamma_i}$.
%\end{quote}

\newcommand{\CLIP}{\textrm{clip}}

%For the initial problem, we use the entirety of each $\Gamma$ as listed at the start of the proof, and each $\INT(v_1)$ encompasses all of $\Gamma_2$. We now evaluate this problem and the resulting sub-problems.
To this end,
let $\CLIP_i(\Gamma_i,\Gamma_2)$ be the sub-arc of $\Gamma_i$ consisting of all edges of $\Gamma_i$ whose supporting lines have angles in $\Lambda(\Gamma_2)+\theta_{p_1p_i}-\theta_{p_1p_2}$ (mod $\pi$).
Let $m(\Gamma_2)=\sum_{i=3}^k|\CLIP_i(\Gamma_i,\Gamma_2)| + |\Gamma_2|$.
Partition $\Gamma_1$ into $r$ sub-arcs such that each sub-arc $\gamma_1$ has $|\Gamma_1|/r \pm O(1)$ edges.
Partition $\Gamma_2$ into $r$ sub-arcs such that each sub-arc $\gamma_2$ has $m(\gamma_2) = m(\Gamma_2)/r \pm O(k)$. (This is possible because for a single edge $e_2$ of $\gamma_2$, $m(e_2) \le O(k)$.)

By Lemma~\ref{lem:int}, we first enumerate all similarity transformations $\varphi$ with $\varphi(p_1)$ on a vertex of $\Gamma_1$, $\varphi(p_2)$ on some sub-arc $\gamma_2$, and $\varphi(p_3)$ on each of the $O(1)$ contiguous pieces of $\Gamma_3 \setminus \CLIP_3(\Gamma_3, \gamma_2)$; the disjointness condition in Lemma~\ref{lem:int} is satisfied by our definition of $\CLIP_i$ (see Figure~\ref{fig:kgon:3contact}).
This gives $O(|\Gamma_1|)$ transformations to check per $\gamma_2$, which requires $\OO(|\Gamma_1|\cdot rk)$ time over all $\gamma_2$ (checking the feasibility of one transformation takes $O(k\log n)$ time,
since we can tell whether any given point is inside $Q$ via binary search \cite{PreparataS}).

\begin{figure}
    \centering    \includegraphics[width=.85\textwidth]{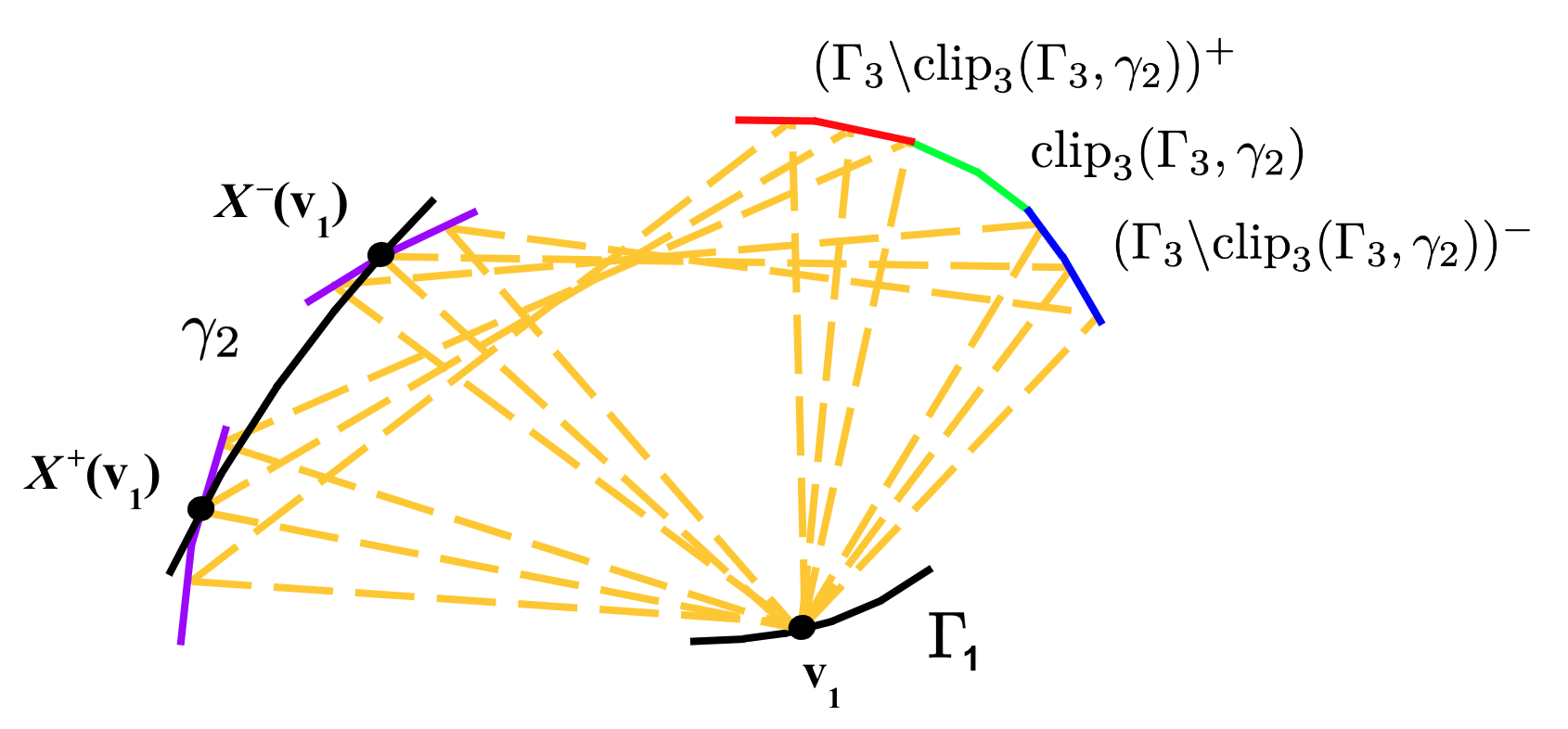}
    \caption{Example where both sub-arcs constituting $\Gamma_3 \setminus \CLIP_3(\Gamma_3, \gamma_2)$ (denoted $(\Gamma_3 \setminus \CLIP_3(\Gamma_3, \gamma_2))^+$ and $(\Gamma_3 \setminus \CLIP_3(\Gamma_3, \gamma_2))^-$) have a similarity transformation that places $p_1$ at a vertex $v_1$ of $\Gamma_1$, $p_2$ on $\gamma_2$, and $p_3$ on $\Gamma_3 \setminus \CLIP_3(\Gamma_3, \gamma_2)$. $X^+(v_1)$ and $X^-(v_2)$ can be found rapidly via Lemma~\ref{lem:int} since $(\Lambda(\Gamma_3 \setminus \CLIP_3(\Gamma_3, \gamma_2)) + \theta_{p_1p_2}) \cap (\Lambda(\gamma_2)+\theta_{p_1p_3}) = \emptyset$ (mod $\pi)$. Dashed triangles are similar to $\triangle p_1p_2p_3$.}
    \label{fig:kgon:3contact}
\end{figure}

It remains to search for transformations $\varphi$ such that $\varphi(p_1)$ is on a vertex of some sub-arc $\gamma_1$, $\varphi(p_2)$ is on some sub-arc $\gamma_2$, and $\varphi(p_3)$ is on $\CLIP_3(\Gamma_3, \gamma_2)$. There are two possibilities:

\begin{itemize}
\item Case 1: For each $i\in\{3,\ldots,k\}$, $\Lambda(\gamma_1)+\theta_{p_2p_i}$ and
$\Lambda(\gamma_2)+\theta_{p_1p_i}$ are disjoint (mod $\pi$). For each $i \in \{3,\ldots,k\}$, we apply Lemma \ref{lem:int} to $(\gamma_1,\gamma_2,\gamma')$ for each of the $O(1)$ contiguous pieces $\gamma'$ of $\Gamma_i\setminus \CLIP_i(\Gamma_i, \gamma_2)$. For each vertex $v_1$ of $\gamma_1$, let $\INT'(v_1)$ be the intersection of all sub-arcs $I(v_1)$ produced during these applications of Lemma~\ref{lem:int}. We now use Lemma \ref{lem:kgon:rs} on $\langle \gamma_1,\gamma_2, \CLIP_3(\Gamma_3, \gamma_2), \ldots, \CLIP_k(\Gamma_k, \gamma_2)\rangle$. The sub-arc we pass to Lemma \ref{lem:kgon:rs} for each vertex $v_1$ of $\gamma_1$ is $\INT(v_1)\cap\INT'(v_1)$.
The total time for all instances of this case is $\OO(r^2\cdot k^2 (|\Gamma_1|/r+m(\Gamma_2)/r+k))$ (since we can operate on a truncated version of $Q$ consisting of just the arcs/sub-arcs specified).

\item Case 2: For some $i\in\{3,\ldots,k\}$, $\Lambda(\gamma_1)+\theta_{p_2p_i}$ and $\Lambda(\gamma_2)+\theta_{p_1p_i}$
intersect (mod $\pi$). Here, we recursively solve the problem for $\langle\gamma_1,\gamma_2,\CLIP_3(\Gamma_3, \gamma_2),\ldots,\CLIP_k(\Gamma_k, \gamma_2)\rangle$.
The sub-arc we pass to the recursive call for each vertex $v_1$ of $\gamma_1$ is $\INT(v_1)\cap\INT'(v_1)$, where $\INT'(v_1)$ is defined as in Case 1.
%$, \INT(v_1^{(1)})\cap\INT'(v_1^{(1)}), \ldots, \INT(v_1^{(|\gamma_1|)})\cap\INT'(v_1^{(|\gamma_1|)}))$.
%Each $\INT'(v_1)$ is produced as in Case 1.
There are $O(r)$ pairs $(\gamma_1,\gamma_2)$ satisfying this condition per $i$, since when we overlay two subdivisions of $\R$ into $O(r)$ intervals, the number of intersecting pairs of intervals is $O(r)$. 
Thus, the total number of recursive calls for this case is $O(kr)$. The time to produce the sub-problems is subsumed by the time bound in Case 1.
\end{itemize}

\newcommand{\mmm}{\widehat{m}}

\noindent
We take the best solution found in all cases.
%This completes the description of our recursive algorithm.
Letting $\mmm=|\Gamma_1|+m(\Gamma_2)$, we obtain the following recurrence for the running time:
%\footnote{We will assume $r \ge \omega(1)$ during the analysis.}
%(note that $m(\Gamma) \le kn$ for all $\Gamma$ and recall that $k \le n$):
\[ T(\mmm)\: \le\: O(kr)\, T(\mmm/r + O(k)) + \OO(k^2r^2 (\mmm/r+k)).
\]
As base case, if $\mmm\le kr$, we use the naive bound $T(\mmm)=O(k^2\mmm^2)$ by constructing the space of all feasible placements \cite{AgarwalAS}.
The recurrence solves to $T(\mmm)\le O(k)^{\log_r \mmm}\cdot (kr)^{O(1)} \mmm$.
Choosing $r=\mmm^\eps$ yields $T(\mmm)\le k^{O(1/\eps)}\mmm^{1+O(\eps)}$. (We can adjust $\eps$ by a constant factor.)
%Letting $n_1=|\Gamma_1|$ and $m=m(\Gamma_2)$, we obtain the following recurrence for the time (note that $m(\Gamma) \le kn$ for all $\Gamma$ and recall that $k \le n$):
%\[ T(n_1,m)\: \le\: O(kr) T(n_1/r + O(1),m/r + O(k)) + \OO(k^2r^2 (n_1/r+m/r+k)).
%\]
%It solves to $T(n_1,m) \le O(k)^{(\log_r(n_1+m))}(n_1+m)\cdot (kr\log_r(n_1+m))^{O(1)}$ if we evaluate the sub-problems of size $\le Ck\log_r(n_1+m)$ (for some sufficiently large constant $C$) by constructing the entire space of feasible placements \cite{AgarwalAS} instead of recursing further.
%Choosing $r=2^{\sqrt{\log k\log n}}$ gives $T(n,kn)\le n2^{O(\sqrt{\log k\log n})}$, which beats $O(k^{O(1/\eps)}n^{1+\eps})$ for any constant $\eps > 0$. 
\end{proof}

Note that our earlier divide-and-conquer approach in Theorem~\ref{thm:tri:3contact} (which yielded a slightly better \\ $O(n\polylog n)$ running time) does not work here. This is because we need to ensure disjointness conditions for multiple triangles $\triangle p_1p_2p_i$ with different rotational shifts, meaning that we cannot use one common interval $S$ to represent the $k$ arcs in a sub-problem.

It is not difficult to modify the algorithm to solve the 2-contact (2-anchor) case, as shown in
\LIPICS{the full paper.}%
\NORMAL{Appendix~\ref{app:2contact}.

}
Extension to the 4-contact case is more challenging, however.
Without an anchor vertex, Lemma~\ref{lem:int} is no longer applicable, so we cannot clip arcs as in Theorem~\ref{thm:kgon:3contact}'s proof. With some care, we could still use the Pairing Lemma to solve the problem, but we would need to pair
some arcs $\Gamma_i$ with $\Gamma_2$ and some arcs $\Gamma_i$ with another arc such as $\Gamma_1$. As a result, the range searching
sub-problems become more complex, and the running time would be much larger (though subquadratic).
In the next section, we suggest a better, more elegant way to solve the 4-contact case in near-linear time,
without needing range searching at all!

\section{4-Contact (No-Anchor) Case}\label{sec:4contact}

To solve the 4-contact case, we will actually show that the number of %(combinatorially distinct\footnote{
%In degenerate scenarios, e.g., when $P$ and $Q$ are squares, there could technically be an infinite number of 4-contact placements; in such cases, we may apply small perturbations, or instead count the number of distinct quadruples of edges of $Q$ corresponding to such placements.
%}) 
solutions (not necessarily optimal nor locally optimal) is actually near-linear in $n$, assuming general position input.\footnote{
In degenerate scenarios, e.g., when $P$ and $Q$ are squares, there could technically be an infinite number of 4-contact placements; in such cases, we may apply small perturbations, or instead count the number of distinct quadruples of edges of $Q$ corresponding to such placements.
}
This allows us to focus on the problem of enumerating all 4-contact placements (as we can check their feasibility rapidly). For the enumeration problem, we can immediately reduce the general $k$ case to the $k=4$ case.

\subsection{Covering All 3-Contact Solutions by Pairs and Triples for $k=3$}

To solve the enumeration problem for $k=4$, we will actually revisit the $k=3$ case.
Although the number of 3-contact solutions may be quadratic in the worst case, we observe that
our divide-and-conquer algorithm from Theorem~\ref{thm:tri:3contact} can generate a near-linear number of pairs that ``cover'' all 3-contact solutions.
To be precise, we make the following definitions:
We say that $(q_1,q_2,q_3)$ is \emph{covered} by a list $L$ of triples of edges if 
$q_1$  is on $e_1$, $q_2$ is on $e_2$, and $q_3$ is on $e_3$ for some $(e_1,e_2,e_3)\in L$.
We say that $(q_1,q_2)$ is \emph{covered} by a pairing $M$ of sub-edges if
$q_1$ is on $e_1$ and $q_2$ is on $M(e_1)$ for some sub-edge~$e_1$.

We begin with a variant of the Pairing Lemma that guarantees monotonically \emph{increasing} pairings, which will be
crucial later%
\LIPICS{ (see the full paper for the proof}:

\begin{lemma}\label{lem:tri:match:modif} \emph{(Modified Pairing Lemma)}\ \
Let $\triangle p_1p_2p_3$ be a triangle.
Let $\Gamma_1,\Gamma_2,\Gamma_3$ be arcs of a convex $n$-gon $Q$, such that

\begin{enumerate}
\item $\Lambda(\Gamma_1)+\theta_{p_2p_3}$ and $\Lambda(\Gamma_2)+\theta_{p_1p_3}$ are disjoint (mod $\pi$), and
\item $\Lambda(\Gamma_1)+\theta_{p_2p_3}$ and $\Lambda(\Gamma_3)+\theta_{p_1p_2}$ are disjoint (mod $\pi$).
\end{enumerate}

\noindent
In $\OO(n)$ time, we can compute a monotonically \emph{increasing} pairing $\MATCH$ between $\Gamma_2$ and $\Gamma_3$ with $O(n)$ sub-edges, or a list $L$ of $O(n)$ triples of edges, satisfying the following property:
\begin{quote}
For every similarity transformation $\varphi$ that has $\varphi(p_1)$ on $\Gamma_1$ and $\varphi(p_2)$ on 
$\Gamma_2$ and $\varphi(p_3)$ on $\Gamma_3$, we have $(\varphi(p_2),\varphi(p_3))$ covered by $\MATCH$ or
$(\varphi(p_1),\varphi(p_2),\varphi(p_3))$ covered by $L$.
\end{quote}
\end{lemma}

\NORMAL{

Before proving Lemma~\ref{lem:tri:match:modif}, we first analyze an easy special case:

\begin{lemma}\label{lem:triple} 
Let $\triangle p_1p_2p_3$ be a triangle.
Let $\Gamma_1,\Gamma_2,\Gamma_3$ be arcs of a convex $n$-gon $Q$, such that

\begin{enumerate}
\item $\Lambda(\Gamma_1)+\theta_{p_2p_3}$ and $\Lambda(\Gamma_2)+\theta_{p_1p_3}$ are disjoint (mod $\pi$), 
\item $\Lambda(\Gamma_1)+\theta_{p_2p_3}$ and $\Lambda(\Gamma_3)+\theta_{p_1p_2}$ are disjoint (mod $\pi$), 
\item $\Lambda(\Gamma_2)+\theta_{p_1p_3}$ and $\Lambda(\Gamma_3)+\theta_{p_1p_2}$ are disjoint (mod $\pi$).
\end{enumerate}

\noindent
In $\OO(n)$ time, we can compute a list $L$ of $O(n)$ triples of edges, satisfying the following property:
\begin{quote}
For every similarity transformation $\varphi$ that has $\varphi(p_1)$ on $\Gamma_1$, $\varphi(p_2)$ on 
$\Gamma_2$, and $\varphi(p_3)$ on $\Gamma_3$, we have $(\varphi(p_1),\varphi(p_2),\varphi(p_3))$
covered by $L$.
%$\varphi(p_1)$ on $e_1$ and $\varphi(p_2)$ on $e_2$
%and $\varphi(p_3)$ on $e_3$ for some $(e_1,e_2,e_3)\in L$.
\end{quote}
\end{lemma}
\begin{proof}
We apply Lemma~\ref{lem:tri:match} twice, to obtain a pairing $\MATCH_{12}$ between $\Gamma_1$ and $\Gamma_2$
and a pairing $\MATCH_{23}$ between $\Gamma_2$ and $\Gamma_3$.
We overlay the 2 subdivisions along $\Gamma_2$.  We can then append the triples $(\MATCH_{12}(e_2),e_2,\MATCH_{23}(e_2))$
over all sub-edges $e_2$ of $\Gamma_2$ to $L$.
\end{proof}

To prove Lemma~\ref{lem:tri:match:modif}, we show that any monotonically decreasing pairing actually implies a division of the arcs into sub-problems that can be processed using Lemma~\ref{lem:triple}:

\begin{proof}
We apply Lemma~\ref{lem:tri:match} to obtain a pairing $\MATCH$, which already satisfies the lemma if $\MATCH$ is monotonically increasing.
We now describe how to convert $\MATCH$ to a list $L$ of triples if it is monotonically decreasing.

Let $\Gamma_2^-(t)$ (resp.\ $\Gamma_2^+(t)$) be the sub-arc of $\Gamma_2$ consisting of all edges whose supporting lines' angles are in the interval $(t-\theta_{p_1p_3} - \pi/2$, $t-\theta_{p_1p_3}]$ (mod $\pi$) (resp. $[t-\theta_{p_1p_3}, t-\theta_{p_1p_3} + \pi/2]$ (mod $\pi$)). $\Gamma_2^-(t)$ and $\Gamma_2^+(t)$ may share at most one edge. The intervals allow wrap-around, so $(a,b]$ indicates $(a,\pi) \cup [0,b]$ if $a > b$.
Let $\Gamma_3^-(t)$ and $\Gamma_3^+(t)$ be defined similarly.
Let $q_2(t)$ denote the dividing point (or shared edge) between $\Gamma_2^-(t)$ and $\Gamma_2^+(t)$, and let $q_3(t)$ denote the dividing point (or shared edge) between $\Gamma_3^-(t)$ and $\Gamma_3^+(t)$.
As $t$ increases, both $q_2(t)$ and $q_3(t)$ are monotonically increasing, i.e., move in ccw order.
Since $\MATCH$ is monotonically decreasing, we can find (via binary search) a value of $t$ such that a sub-edge incident to (or contained within) $q_2(t)$ is paired with a sub-edge incident to (or contained within) $q_3(t)$. Let us fix this value of $t$.

As long as at least one of $q_2(t)$ or $q_3(t)$ is not an edge, we have that $\Lambda(\Gamma_2^-(t))+\theta_{p_1p_3}$ and $\Lambda(\Gamma_3^+(t))+\theta_{p_1p_2}$ are disjoint (mod $\pi$). In this case,
we can apply Lemma~\ref{lem:triple} to $(\Gamma_1,\Gamma_2^-(t),\Gamma_3^+(t))$ to get $O(n)$ triples.
Similarly, we can apply Lemma~\ref{lem:triple} to $(\Gamma_1,\Gamma_2^+(t),\Gamma_3^-(t))$ to get $O(n)$ triples.
On the other hand, in $\MATCH$, there are no pairs of sub-edges in $\Gamma_2^-(t)$  with sub-edges in $\Gamma_3^-(t)$,
nor pairs of sub-edge in $\Gamma_2^+(t)$ with sub-edges in $\Gamma_3^+(t)$.

If $q_2(t)$ and $q_3(t)$ are both edges, then there exists exactly one sub-edge $e_2$ of $q_2(t)$ paired with a sub-edge $e_3$ of $q_3(t)$. We enumerate all $O(n)$ triples of the form $(e_1, q_2(t), q_3(t))$ and append these to $L$. Then, we apply Lemma~\ref{lem:triple} to $(\Gamma_1,\Gamma_2^-(t)\setminus  e_2,\Gamma_3^+(t)\setminus e_3)$ and $(\Gamma_1,\Gamma_2^+(t)\setminus e_2,\Gamma_3^-(t)\setminus e_3)$ as before. Note that one or more of $\Gamma_2^-(t)\setminus e_2$, $\Gamma_2^+(t)\setminus e_2$, $\Gamma_3^-(t)\setminus e_3$, or $\Gamma_3^+(t)\setminus e_3$ might consist of two contiguous pieces instead of one; this can be handled by applying Lemma~\ref{lem:triple} to all $O(1)$ combinations of contiguous pieces.
\end{proof}

}

\newcommand{\MMM}{{\cal M}}

By adapting our $k=3$ algorithm, we can cover all 3-contact solutions
by $O(\log n)$ monotonically increasing pairings, together with an extra set of $O(n\log n)$ triples:

\begin{theorem}\label{thm:tri:cover} 
Let $\triangle p_1p_2p_3$ be a triangle.
Let $\Gamma_1,\Gamma_2,\Gamma_3$ be arcs of a convex $n$-gon $Q$.
In $\OO(n)$ time, we can compute a collection $\MMM$ of
$O(\log n)$ monotonically \emph{increasing} pairings between $\Gamma_1$ and $\Gamma_2$, between $\Gamma_2$ and $\Gamma_3$,
and between $\Gamma_1$ and $\Gamma_3$, each with $O(n)$ sub-edges, and a list $L$ of $O(n\log n)$ triples of edges, satisfying
the following property:
\begin{quote}
For every similarity transformation $\varphi$ that has $\varphi(p_1)$ on $\Gamma_1$ and $\varphi(p_2)$ on 
$\Gamma_2$ and $\varphi(p_3)$ on $\Gamma_3$, we have $(\varphi(p_1),\varphi(p_2))$ or $(\varphi(p_2),\varphi(p_3))$ or
$(\varphi(p_1),\varphi(p_3))$ covered by some pairing in $\MMM$, or $(\varphi(p_1),\varphi(p_2),\varphi(p_3))$ covered by $L$.
\end{quote}
\end{theorem}
\begin{proof}
We modify the divide-and-conquer algorithm in the proof of Theorem~\ref{thm:tri:3contact}.
Let $S$ be an interval.
Define $\Gamma_1(S),\Gamma_2(S),\Gamma_3(S)$ as before.
We will recursively solve the problem for $\Gamma_1(S),\Gamma_2(S),\Gamma_3(S)$.

As a first step, we remove edges not participating in $\Gamma_1(S),\Gamma_2(S),\Gamma_3(S)$, so that  the number of edges in $Q$ 
reduced to $m(S):=|\Gamma_1(S)|+|\Gamma_2(S)|+|\Gamma_3(S)|$.
Divide $S$ into two disjoint sub-intervals $S^-$ and $S^+$ so that $m(S^-),m(S^+) = m(S)/2 \pm O(1)$.
We consider various possibilities:

\begin{itemize}
\item Case 1: $\varphi(p_1)$ is on $\Gamma_1(S^-)$, $\varphi(p_2)$ is on $\Gamma_2(S^-)$, and $\varphi(p_3)$ is on $\Gamma_3(S^-)$.
We can recursively solve the problem for $S^-$.
\item Case 2: $\varphi(p_1)$ is on $\Gamma_1(S^+)$, $\varphi(p_2)$ is on $\Gamma_2(S^+)$, and $\varphi(p_3)$ is on $\Gamma_3(S^+)$.
We can recursively solve the problem for $S^+$.
\item Case 3: $\varphi(p_1)$ is on $\Gamma_1(S^-)$, $\varphi(p_2)$ is on $\Gamma_2(S^+)$, and $\varphi(p_3)$ is on $\Gamma_3(S^+)$.
Since $\Lambda(\Gamma_1(S^-))+\theta_{p_2p_3}\subseteq S^-$ and $\Lambda(\Gamma_2(S^+))+\theta_{p_1p_3}\subseteq S^+$ are disjoint,
and $\Lambda(\Gamma_1(S^-))+\theta_{p_2p_3}\subseteq S^-$ and $\Lambda(\Gamma_3(S^+))+\theta_{p_1p_2}\subseteq S^+$ are disjoint (mod $\pi$),
we can solve the problem by Lemma \ref{lem:tri:match:modif} in $\OO(m(S))$ time.
\end{itemize}

\noindent
All remaining cases are symmetric to Case 3.  (The previous Case 4 is now symmetric to Case 3, since $p_1$ is no longer treated as a special
anchor vertex.)

A pairing between $\Gamma_1(S^-)$ and $\Gamma_2(S^-)$ and a pairing between $\Gamma_1(S^+)$ and $\Gamma_2(S^+)$ produced by the recursive calls in Cases 1 and 2 can be joined into one pairing while
remaining monotonically increasing, since $\Gamma_1(S^-)$ precedes $\Gamma_1(S^+)$ and $\Gamma_2(S^-)$ precedes $\Gamma_2(S^+)$ in
ccw order. We can join the pairings for $\Gamma_2,\Gamma_3$ and $\Gamma_1,\Gamma_3$ similarly.
%Similarly, the pairings between $\Gamma_2$ and $\Gamma_3$, and between $\Gamma_1$ and $\Gamma_3$, can be joined.
(And we can trivially union the lists of triples together.)

This yields a total of $O(\log n)$ pairings each with $O(n)$ sub-edges, plus an extra list of $O(n\log n)$ triples.
\end{proof}

\subsection{Enumerating All 4-Contact Solutions for $k=4$}

To solve the enumeration problem for $k=4$, we claim that we do not need any further ingredients!  We can just
run our $k=3$ algorithm for each of the 4 triangles from the input 4-gon and then piece the outputs together in a careful way.

\begin{lemma}\label{lem:3}
Let $p_1p_2p_3p_4$ be a $4$-gon and $Q$ be a convex $n$-gon in general position.
Given 3 edges $e_1,e_2,e_3$ of $Q$, there are only $O(1)$ %edges $e_4$ with a 
similarity transformations $\varphi$ with $\varphi(p_1)$  on $e_1$,
$\varphi(p_2)$  on $e_2$, $\varphi(p_3)$ on $e_3$, and $\varphi(p_4)$ on $\partial Q$,
and they can be computed in $O(\log n)$ time.
\end{lemma}
\LIPICS{
\begin{proof}
See the full paper.
\end{proof}
}
\NORMAL{
\begin{proof}
%Recall that $\varphi_{s,t,u,v}:\R^2\rightarrow\R^2$ denotes the similarity transformation $(x,y)\mapsto 
%(sx-ty+u,tx+sy+v)$. 
Let $\varphi_{s,t,u,v}$ be as before.
Define 
\[ \mathcal{L} = \{(s,t,u,v)\in\R^4: \mbox{$\varphi_{s,t,u,v}(p_1)$ is on $e_1$, $\varphi_{s,t,u,v}(p_2)$ is on $e_2$, and
$\varphi_{s,t,u,v}(p_3)$ is on $e_3$}\}.
\]
Then $\mathcal{L}$ is a line segment in $\R^4$  (since there are 3 linear equality constraints and 6 inequality constraints in the variables $s,t,u,v$). Thus, $\mathcal{L}'=\{\varphi_{s,t,u,v}(p_4): (s,t,u,v)\in \mathcal{L}\}$ is a line segment in $\R^2$.  The edges we want correspond to intersections of $\mathcal{L}'$ with $\partial Q$. There are at most 2 such intersections, and they can be found via binary search~\cite{PreparataS}.
\end{proof}
}

\begin{theorem}\label{thm:4gon}
Let $P$ be a $4$-gon with vertices $p_1,p_2,p_3,p_4$ and $Q$ be a convex $n$-gon in general position.
There are at most $O(n\log^2 n)$ %(combinatorially distinct) 
similarity transformations $\varphi$ such that $\varphi(p_1),\varphi(p_2),\varphi(p_3),\varphi(p_4)$
are on $\partial Q$, and they can be enumerated in $\OO(n)$ time.
\end{theorem}
\begin{proof}
Divide $\partial Q$ into $O(1)$ arcs, and 
let $\Gamma_1,\ldots,\Gamma_4$ be 4 such arcs of $Q$.
(We will try all $O(1)$ choices for $\Gamma_1,\ldots,\Gamma_4$.)
We apply Theorem~\ref{thm:tri:cover} to the 4 triangles $\triangle p_1p_2p_3$, $\triangle p_1p_2p_4$, $\triangle p_1p_3p_4$, and
$\triangle p_2p_3p_4$,
to get a combined collection $\MMM$ of $O(\log n)$ monotonically increasing pairings and 
a combined list $L$ of triples.

We consider various possibilities (we will try them all and return the union of the outputs).  If $(\varphi(p_1),\varphi(p_2),\\ \varphi(p_3))$,
$(\varphi(p_1),\varphi(p_2),\varphi(p_4))$, $(\varphi(p_1),\varphi(p_3),\varphi(p_4))$, or $(\varphi(p_2),\varphi(p_3),\varphi(p_4))$
is covered by $L$, we can examine each of the $O(n\log n)$ triples in $L$, and use Lemma~\ref{lem:3} to generate $O(1)$ transformations $\varphi$ per triple, in $\OO(n)$ time.
%this gives $O(n\log n)$ 4-contact placements.

Otherwise, define a small graph $G_\varphi$ with vertices $\{1,2,3,4\}$, where $ij$ is an edge iff $(\varphi(p_i),\varphi(p_j))$
is covered by a pairing in $\MMM$.  We know that each of triple of vertices contains an edge in $G_\varphi$.  It is easy to see (from a short case analysis) that $G_\varphi$ must have $\ge2$ edges.

%Otherwise, consider an arbitrary similarity transformation $\varphi$ that gives a 4-contact placement not covered using $L$. Define a small graph $G$ with vertices $\{1,2,3,4\}$, where $ij$ is an edge iff $(\varphi(p_i),\varphi(p_j))$ is covered by a pairing in $\MMM$. We know that each triple of vertices in $G$ contains at least one edge. It is easy to see (from a short case analysis) that $G$ must have at least 2 edges. Therefore, we can handle all possible remaining 4-contact placements by checking for the following cases.

\begin{itemize}
\item Case 1: $G_\varphi$ contains 2 adjacent edges, w.l.o.g., 12 and 23.
Then $(\varphi(p_1),\varphi(p_2))$ is covered by a pairing $\MATCH_{12}\in\MMM$ between $\Gamma_1$ and $\Gamma_2$,
and $(\varphi(p_2),\varphi(p_3))$ is covered by a pairing $\MATCH_{23}\in\MMM$ between $\Gamma_2$ and $\Gamma_3$.
We overlay the 2 subdivisions along $\Gamma_2$.  We examine the triple $(\MATCH_{12}(e_2),e_2,\MATCH_{23}(e_2))$
for each sub-edge $e_2$ of $\Gamma_2$, and  use Lemma~\ref{lem:3} to generate $O(1)$ transformations $\varphi$ per triple, in $\OO(n)$ time.  The total number of resulting triples
over all $O(\log^2n)$ choices of $\MATCH_{12}$ and $\MATCH_{23}$ is $O(n\log^2n)$, giving $O(n\log^2n)$ transformations.
%4-contact placements.

\vspace{5px}
\item Case 2: $G_\varphi$ contains 2 independent edges, w.l.o.g., 12 and 34.
Then $(\varphi(p_1),\varphi(p_2))$ is covered by a pairing $M_{12}\in\MMM$ between $\Gamma_1$ and $\Gamma_2$,
and $(\varphi(p_3),\varphi(p_4))$ is covered by a pairing $M_{34}\in\MMM$ between $\Gamma_3$ and $\Gamma_4$.

\vspace{5px}
For 2 sub-edges $e$ and $e'$, define the angle interval $\Theta(e,e')=\{\theta_{qq'}: q\in e,\ q'\in e'\}$.
It suffices to enumerate quadruples $(e_1,M_{12}(e_1),e_3,M_{34}(e_3))$ over all sub-edges $e_1$ of $\Gamma_1$
and all sub-edges $e_3$ of $\Gamma_3$, under the restriction that $\Theta(e_1,M_{12}(e_1))-\theta_{p_1p_2}$ intersects
$\Theta(e_3,M_{34}(e_3))-\theta_{p_3p_4}$ (mod $\pi$). See Figure~\ref{fig:4gon} for an example quadruple.

\vspace{5px}
Observe that because $M_{12}$ is monotonically increasing, the angle intervals $\Theta(e_1,M_{12}(e_1))$ are disjoint\footnote{$\Theta(e_1,M_{12}(e_1))$ and $\Theta(e_1',M_{12}(e_1'))$ may share a limit point if $e_1$ and $e_1'$ are adjacent, but this does not affect the proof.} and move monotonically as $e_1$ moves in ccw order. 
Similarly, because $M_{34}$ is monotonically increasing, the angle intervals $\Theta(e_3,M_{34}(e_3))$ are disjoint
and move monotonically as $e_3$ moves in ccw order. 
By overlaying the 2 sets of $O(n)$ intervals $\{\Theta(e_1,M_{12}(e_1))-\theta_{p_1p_2}: \text{ sub-edge } e_1 \text{ of } \Gamma_1\}$
and  $\{\Theta(e_3,M_{34}(e_3))-\theta_{p_3p_4}: \text{ sub-edge } e_3 \text{ of } \Gamma_3\}$, we see that
there are at most $O(n)$ choices of $(e_1,e_3)$ such that $\Theta(e_1,M_{12}(e_1))-\theta_{p_1p_2}$ intersects
$\Theta(e_3,M_{34}(e_3))-\theta_{p_3p_4}$ (mod $\pi$), and they can be enumerated in $O(n)$ time.  The total number
of quadruples over all $O(\log^2n)$ choices of $\MATCH_{12}$ and $\MATCH_{34}$ is $O(n\log^2n)$, giving $O(n\log^2n)$ transformations.
\end{itemize}
\vspace{-1.35em}
\end{proof}

Interestingly, Case 2 exploits a different phenomenon than in our earlier proofs: besides monotone pairings of sub-edges, we have monotone ``pairings of pairs'' of sub-edges.

\begin{figure}
    \centering
    \includegraphics[width=.5\textwidth]{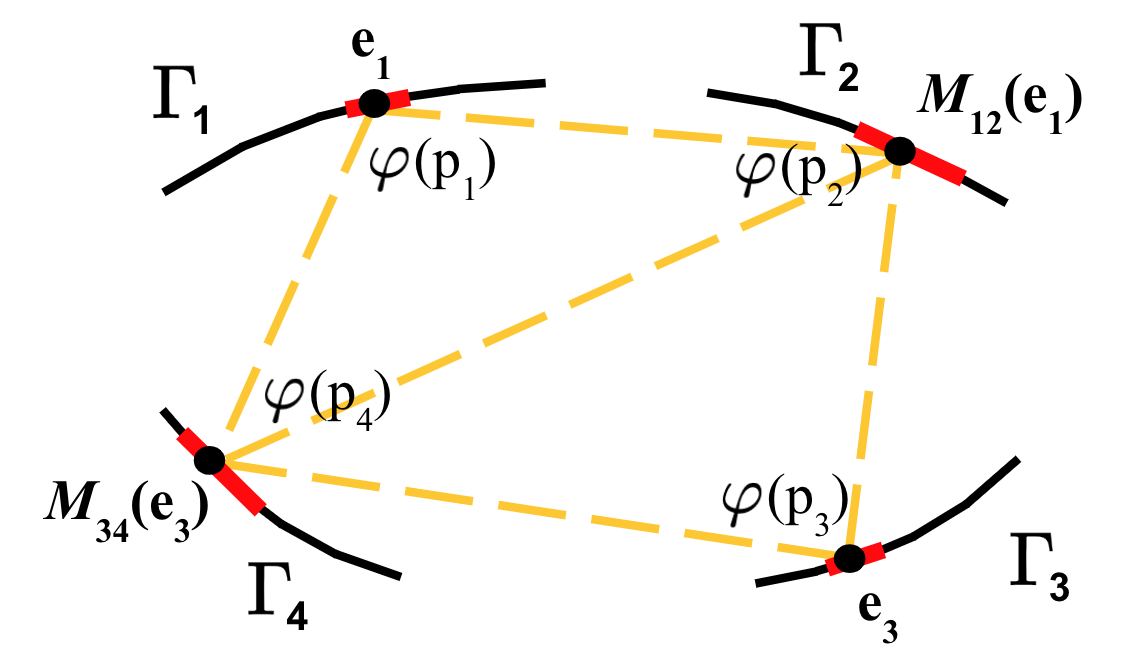}
    \caption{An example quadruple from Case 2 in the proof of Theorem~\ref{thm:4gon}. The pairing $M_{12}$ is produced using $\triangle p_1p_2p_4$ and the pairing $M_{34}$ is produced using $\triangle p_2p_3p_4$.}
    \label{fig:4gon}
\end{figure}

\subsection{Generalizing to $k>4$}

Finally, an algorithm for the 4-contact case for general $k$ immediately follows:

\begin{corollary}\label{cor:kgon:4contact}
Given a $k$-gon $P$ and a convex $n$-gon $Q$ in general position,
there are at most $O(k^4n\log^2 n)$ %combinatorially distinct 
similar copies of $P$ contained in $Q$ that have $4$ different vertices of $P$ on $4$ edges of $Q$,
and they can be enumerated in  $\OO(k^5 n)$ time.
\end{corollary}
\begin{proof}
For each of the $O(k^4)$ choices of vertices $p_1,p_2,p_3,p_4$ of $P$,
we generate the $O(n\log^2n)$ similarity transformations from Theorem~\ref{thm:4gon} and test the feasibility
of each transformation in $O(k\log n)$ time via binary searches \cite{BergCKO08}.
\end{proof}

\NORMAL{
\section{Final Remarks}

To keep the presentation cleaner, we have not spelled out the precise polylogarithmic factors in the time bounds,
nor the precise dependencies on $k$ in the $k^{O(1/\eps)}$ factor, since we believe that extra effort
can help in optimizing those factors. Even without extra effort, if we use randomized search~\cite{Chan99}
in the proof of Lemma~\ref{lem:tri:rs}, the number of logarithmic factors in our $O(n\polylog n)$-time algorithm for $k=3$ is only about~3.

The main question we leave open is whether Problem~\ref{prob} can be solved in $O(kn^{1+\varepsilon})$ time for all constant $\varepsilon > 0$, which would be a strict improvement over the previous $\OO(kn^2)$ bound for all $k$.
Similarly, could the number of 4-contact placements of $P$ within $Q$ be upper-bounded by $\OO(kn)$ instead of $\OO(k^4n)$?

}

%\nocite{*}
\bibliographystyle{plainurl}
\bibliography{refs}

\begin{thebibliography}{10}

\bibitem{AgarwalAS}
Pankaj~K. Agarwal, Nina Amenta, and Micha Sharir.
\newblock {L}argest {P}lacement of {O}ne {C}onvex {P}olygon {I}nside {A}nother.
\newblock {\em Discret. Comput. Geom.}, 19(1):95--104, 1998.
\newblock \href {https://doi.org/10.1007/PL00009337}
  {\path{doi:10.1007/PL00009337}}.

\bibitem{AgarwalAS99}
Pankaj~K. Agarwal, Boris Aronov, and Micha Sharir.
\newblock {M}otion {P}lanning for a {C}onvex {P}olygon in a {P}olygonal
  {E}nvironment.
\newblock {\em Discret. Comput. Geom.}, 22(2):201--221, 1999.
\newblock \href {https://doi.org/10.1007/PL00009455}
  {\path{doi:10.1007/PL00009455}}.

\bibitem{AgarwalEricksonSURVEY}
Pankaj~K. Agarwal and Jeff Erickson.
\newblock {G}eometric {R}ange {S}earching and {I}ts {R}elatives.
\newblock {\em Contemporary Mathematics}, 223:1--56, 1999.
\newblock URL: \url{https://jeffe.cs.illinois.edu/pubs/pdf/survey.pdf}.

\bibitem{SharirBOOK}
Pankaj~K. Agarwal and Micha Sharir.
\newblock {D}avenport-{S}chinzel {S}equences and {T}heir {G}eometric
  {A}pplications.
\newblock In J{\"{o}}rg{-}R{\"{u}}diger Sack and Jorge Urrutia, editors, {\em
  Handbook of Computational Geometry}, pages 1--47. North Holland / Elsevier,
  2000.
\newblock \href {https://doi.org/10.1016/B978-044482537-7/50002-4}
  {\path{doi:10.1016/B978-044482537-7/50002-4}}.

\bibitem{Aronov23}
Boris Aronov, Mark de~Berg, Jean Cardinal, Esther Ezra, John Iacono, and Micha
  Sharir.
\newblock {S}ubquadratic {A}lgorithms for {S}ome 3{SUM}-{H}ard {G}eometric
  {P}roblems in the {A}lgebraic {D}ecision-{T}ree {M}odel.
\newblock {\em Comput. Geom.}, 109:101945, 2023.
\newblock \href {https://doi.org/10.1016/J.COMGEO.2022.101945}
  {\path{doi:10.1016/J.COMGEO.2022.101945}}.

\bibitem{AvBoissonnat}
Francis Avnaim and Jean{-}Daniel Boissonnat.
\newblock Polygon {P}lacement {U}nder {T}ranslation and {R}otation.
\newblock {\em {RAIRO} Theor. Informatics Appl.}, 23(1):5--28, 1989.
\newblock \href {https://doi.org/10.1051/ITA/1989230100051}
  {\path{doi:10.1051/ITA/1989230100051}}.

\bibitem{Baird}
Henry~Spalding Baird.
\newblock {\em {M}odel-{B}ased {I}mage {M}atching {U}sing {L}ocation}.
\newblock MIT Press, 1984.

\bibitem{Barba19}
Luis Barba, Jean Cardinal, John Iacono, Stefan Langerman, Aur{\'{e}}lien Ooms,
  and Noam Solomon.
\newblock {S}ubquadratic {A}lgorithms for {A}lgebraic 3{SUM}.
\newblock {\em Discret. Comput. Geom.}, 61(4):698--734, 2019.
\newblock Preliminary version in SoCG 2017.
\newblock \href {https://doi.org/10.1007/S00454-018-0040-Y}
  {\path{doi:10.1007/S00454-018-0040-Y}}.

\bibitem{BarequetHarPeled}
Gill Barequet and Sariel Har{-}Peled.
\newblock {P}olygon {C}ontainment and {T}ranslational
  {M}in-{H}ausdorff-{D}istance {B}etween {S}egment {S}ets are 3{SUM}-{H}ard.
\newblock {\em Int. J. Comput. Geom. Appl.}, 11(4):465--474, 2001.
\newblock \href {https://doi.org/10.1142/S0218195901000596}
  {\path{doi:10.1142/S0218195901000596}}.

\bibitem{Chan99}
Timothy~M. Chan.
\newblock {G}eometric {A}pplications of a {R}andomized {O}ptimization
  {T}echnique.
\newblock {\em Discret. Comput. Geom.}, 22(4):547--567, 1999.
\newblock Preliminary version in SoCG 1998.
\newblock \href {https://doi.org/10.1007/PL00009478}
  {\path{doi:10.1007/PL00009478}}.

\bibitem{Chan18}
Timothy~M. Chan.
\newblock {M}ore {L}ogarithmic-{F}actor {S}peedups for 3{SUM}, (median,
  +)-{C}onvolution, and {S}ome {G}eometric 3{SUM}-{H}ard {P}roblems.
\newblock {\em {ACM} Trans. Algorithms}, 16(1):7:1--7:23, 2020.
\newblock \href {https://doi.org/10.1145/3363541} {\path{doi:10.1145/3363541}}.

\bibitem{Chazelle}
Bernard Chazelle.
\newblock {T}he {P}olygon {C}ontainment {P}roblem.
\newblock {\em Advances in Computing Research}, 1(1):1--33, 1983.
\newblock URL:
  \url{https://www.cs.princeton.edu/~chazelle/pubs/PolygContainmentProb.pdf}.

\bibitem{ChewKedem}
L.~Paul Chew and Klara Kedem.
\newblock {A} {C}onvex {P}olygon {A}mong {P}olygonal {O}bstacles: {P}lacement
  and {H}igh-{C}learance {M}otion.
\newblock {\em Comput. Geom.}, 3:59--89, 1993.
\newblock \href {https://doi.org/10.1016/0925-7721(93)90001-M}
  {\path{doi:10.1016/0925-7721(93)90001-M}}.

\bibitem{DanielsM97}
Karen~L. Daniels and Victor Milenkovic.
\newblock {M}ultiple {T}ranslational {C}ontainment. {P}art {{I}:} {A}n
  {A}pproximate {A}lgorithm.
\newblock {\em Algorithmica}, 19(1/2):148--182, 1997.
\newblock \href {https://doi.org/10.1007/PL00014415}
  {\path{doi:10.1007/PL00014415}}.

\bibitem{BergCKO08}
Mark de~Berg, Otfried Cheong, Marc~J. van Kreveld, and Mark~H. Overmars.
\newblock {\em {C}omputational {G}eometry: {A}lgorithms and {A}pplications}.
\newblock Springer, 3rd edition, 2008.
\newblock URL: \url{https://www.worldcat.org/oclc/227584184}.

\bibitem{DickersonS96}
Matthew Dickerson and Daniel Scharstein.
\newblock {O}ptimal {P}lacement of {C}onvex {P}olygons to {M}aximize {P}oint
  {C}ontainment.
\newblock In {\em Proc. 7th {ACM-SIAM} Symposium on Discrete Algorithm (SODA)},
  pages 114--121, 1996.
\newblock URL: \url{http://dl.acm.org/citation.cfm?id=313852.313899}.

\bibitem{EomLeeAhn}
Taekang Eom, Seungjun Lee, and Hee{-}Kap Ahn.
\newblock {L}argest {S}imilar {C}opies of {C}onvex {P}olygons {A}midst
  {P}olygonal {O}bstacles.
\newblock {\em CoRR}, abs/2012.06978, 2020.
\newblock \href {http://arxiv.org/abs/2012.06978} {\path{arXiv:2012.06978}}.

\bibitem{Gronlund14}
Allan Gr{\o}nlund and Seth Pettie.
\newblock Threesomes, {D}egenerates, and {L}ove {T}riangles.
\newblock {\em J. {ACM}}, 65(4):22:1--22:25, 2018.
\newblock \href {https://doi.org/10.1145/3185378} {\path{doi:10.1145/3185378}}.

\bibitem{Kane18}
Daniel~M. Kane, Shachar Lovett, and Shay Moran.
\newblock {N}ear-{O}ptimal {L}inear {D}ecision {T}rees for {$k$}-{SUM} and
  {R}elated {P}roblems.
\newblock {\em J. {ACM}}, 66(3):16:1--16:18, 2019.
\newblock \href {https://doi.org/10.1145/3285953} {\path{doi:10.1145/3285953}}.

\bibitem{KunnemannNusser}
Marvin K{\"{u}}nnemann and Andr{\'{e}} Nusser.
\newblock {P}olygon {P}lacement {R}evisited: ({D}egree of {F}reedom + 1)-{SUM}
  {H}ardness and an {I}mprovement via {O}ffline {D}ynamic {R}ectangle {U}nion.
\newblock In {\em Proc. {ACM-SIAM} Symposium on Discrete Algorithms (SODA)},
  pages 3181--3201, 2022.
\newblock \href {https://doi.org/10.1137/1.9781611977073.124}
  {\path{doi:10.1137/1.9781611977073.124}}.

\bibitem{Lee04}
D.~T. Lee.
\newblock {I}nterval, {S}egment, {R}ange, and {P}riority {S}earch {T}rees.
\newblock In Dinesh~P. Mehta and Sartaj Sahni, editors, {\em Handbook of Data
  Structures and Applications}. Chapman and Hall/CRC, 2004.
\newblock \href {https://doi.org/10.1201/9781420035179.CH18}
  {\path{doi:10.1201/9781420035179.CH18}}.

\bibitem{LeeEomAhn}
Seungjun Lee, Taekang Eom, and Hee{-}Kap Ahn.
\newblock {L}argest {T}riangles in a {P}olygon.
\newblock {\em Comput. Geom.}, 98:101792, 2021.
\newblock \href {https://doi.org/10.1016/J.COMGEO.2021.101792}
  {\path{doi:10.1016/J.COMGEO.2021.101792}}.

\bibitem{Megiddo}
Nimrod Megiddo.
\newblock {A}pplying {P}arallel {C}omputation {A}lgorithms in the {D}esign of
  {S}erial {A}lgorithms.
\newblock {\em J. {ACM}}, 30(4):852--865, 1983.
\newblock \href {https://doi.org/10.1145/2157.322410}
  {\path{doi:10.1145/2157.322410}}.

\bibitem{Milenkovic96}
Victor Milenkovic.
\newblock {T}ranslational {P}olygon {C}ontainment and {M}inimal {E}nclosure
  {U}sing {L}inear {P}rogramming {B}ased {R}estriction.
\newblock In {\em Proc. 28th {ACM} Symposium on Theory of Computing (STOC)},
  pages 109--118, 1996.
\newblock \href {https://doi.org/10.1145/237814.237840}
  {\path{doi:10.1145/237814.237840}}.

\bibitem{Milenkovic97}
Victor Milenkovic.
\newblock {M}ultiple {T}ranslational {C}ontainment. {P}art {{II}:} {E}xact
  {A}lgorithms.
\newblock {\em Algorithmica}, 19(1/2):183--218, 1997.
\newblock \href {https://doi.org/10.1007/PL00014416}
  {\path{doi:10.1007/PL00014416}}.

\bibitem{Milenkovic}
Victor Milenkovic.
\newblock {R}otational {P}olygon {C}ontainment and {M}inimum {E}nclosure
  {U}sing {O}nly {R}obust 2{D} {C}onstructions.
\newblock {\em Comput. Geom.}, 13(1):3--19, 1999.
\newblock \href {https://doi.org/10.1016/S0925-7721(99)00006-1}
  {\path{doi:10.1016/S0925-7721(99)00006-1}}.

\bibitem{ORourkeSuriTothPOLYGONS}
Joseph O’Rourke, Subhash Suri, and Csaba~D T{\'o}th.
\newblock {P}olygons.
\newblock In {\em Handbook of Discrete and Computational Geometry}, pages
  787--810. CRC Press, 2017.
\newblock URL: \url{http://www.csun.edu/~ctoth/Handbook/chap30.pdf}.

\bibitem{PreparataS}
Franco~P. Preparata and Michael~Ian Shamos.
\newblock {\em {C}omputational {G}eometry: {A}n {I}ntroduction}.
\newblock Texts and Monographs in Computer Science. Springer, 1985.
\newblock \href {https://doi.org/10.1007/978-1-4612-1098-6}
  {\path{doi:10.1007/978-1-4612-1098-6}}.

\bibitem{Sharir96}
Micha Sharir.
\newblock {A} {N}ear-{L}inear {A}lgorithm for the {P}lanar 2-{C}enter
  {P}roblem.
\newblock {\em Discret. Comput. Geom.}, 18(2):125--134, 1997.
\newblock Preliminary version in SoCG 1996.
\newblock \href {https://doi.org/10.1007/PL00009311}
  {\path{doi:10.1007/PL00009311}}.

\bibitem{SharirT}
Micha Sharir and Sivan Toledo.
\newblock {E}xternal {P}olygon {C}ontainment {P}roblems.
\newblock {\em Computational Geometry}, 4(2):99--118, 1994.
\newblock Preliminary version in SoCG 1991.
\newblock \href {https://doi.org/10.1016/0925-7721(94)90011-6}
  {\path{doi:10.1016/0925-7721(94)90011-6}}.

\end{thebibliography}

\NORMAL{%%%%%%%%%%%%%%%%

\appendix

\section{2-Contact (2-Anchor) Case}\label{app:2contact}

We note that the ideas behind our algorithms in Section~\ref{sec:3contact} for the 3-contact case can be adapted to handle the 2-contact case as well.  In what follows, suppose that the 2 anchor vertices are $p_1$ and $p_2$.

\subsection{Algorithm for $k=3$}

\newcommand{\PIE}[1]{#1^{\bullet}}

We modify the divide-and-conquer algorithm for $k=3$ in Theorem~\ref{thm:tri:3contact}.
The algorithm actually becomes simpler, as we only need Lemma \ref{lem:int} but not Lemma \ref{lem:tri:match}, and the range searching sub-problems are simpler (no need for ellipses or parametric search). 

W.l.o.g., assume that the origin $o$ is in the interior of $Q$.  For an arc $\Gamma$ of $Q$ which is delimited by points $u$ and $v$, let $\PIE{\Gamma}$ denote the subpolygon bounded by $\Gamma$ and the two rays $\overrightarrow{ou}$ and $\overrightarrow{ov}$.

%The input to the recursive algorithm is an interval $S$, together with a sub-arc $I(v_1)\subseteq \Gamma_2(S)$ for every vertex $v_1$ of $\Gamma_1(S)$, and a sub-arc $I(v_2)\subseteq \Gamma_1(S)$ for every vertex $v_2$ of $\Gamma_2(S)$.
Given an interval $S$, the algorithm will find a similarity transformation $\varphi$, maximizing the scaling factor, such that
$\varphi(p_1)$ is a vertex $v_1$ of $\Gamma_1(S)$,
$\varphi(p_2)$ is a vertex $v_2$ of $\Gamma_2(S)$,
%$v_1$ is on $I(v_2)$, $v_2$ is on $I(v_1)$, 
and $\varphi(p_3)$ is inside $\PIE{\Gamma_3(S)}$.

As before, we partition $S$ into $S^-$ and $S^+$. 
The cases are now as follows:

\begin{itemize}
\item Case 1: $\varphi(p_1)$ is on $\Gamma_1(S^-)$, $\varphi(p_2)$ is on $\Gamma_2(S^-)$, and $\varphi(p_3)$ is inside $\PIE{\Gamma_3(S^-)}$.
We can recursively solve the problem for $S^-$.
\item Case 2: $\varphi(p_1)$ is on $\Gamma_1(S^+)$, $\varphi(p_2)$ is on $\Gamma_2(S^+)$, and $\varphi(p_3)$ is inside $\PIE{\Gamma_3(S^+)}$.
We can recursively solve the problem for $S^+$.
\item Case 3: $\varphi(p_2)$ is on $\Gamma_2(S^-)$, and $\varphi(p_3)$ is inside $\PIE{\Gamma_3(S^+)}$.
Since $\Lambda(\Gamma_2(S^-))+\theta_{p_1p_3}\subseteq S^-$ and $\Lambda(\Gamma_3(S^+))+\theta_{p_1p_2}\subseteq S^+$ are disjoint (mod $\pi$),
we can use Lemma \ref{lem:int} to find a 
sub-arc $I(v_1)\subseteq\Gamma_2(S)$ for every vertex $v_1$ of $\Gamma_1(S)$, such that the condition that  $\varphi(p_3)$ is inside $\PIE{\Gamma_3(S^+)}$ is equivalent to the condition that $\varphi(p_2)$ is on $I(v_1)$.  (Technically, Lemma \ref{lem:int} works with $\ext{\Gamma_3(S^+)}$ instead of $\PIE{\Gamma_3(S^+)}$, but we can just intersect the sub-arc $I(v_1)$ with 2 additional halfplanes that arise from the 2 rays bounding $\PIE{\Gamma_3(S^+)}$.)

The problem in this case now reduces to finding a vertex $v_1$ of $\Gamma_1(S)$ and
a vertex $v_2$ of $\Gamma_2(S^-)$,
maximizing their Euclidean distance,
such that 
$v_2$ is on $I(v_1)$.
This further reduces to answering farthest neighbor queries for a 2D point set with an additional 1D interval constraint.  Farthest neighbor queries in 2D can be answered $\OO(1)$ time after  preprocessing in near-linear time (by point location in the farthest-point Voronoi diagram~\cite{BergCKO08,PreparataS}).  The additional 1D interval constraints can be handled by multi-leveling with range trees~\cite{AgarwalEricksonSURVEY,Lee04}
as before, which increases time bounds by a logarithmic factor.
\end{itemize}

\noindent
All remaining cases are symmetric to Case~3
(swapping subscripts 1 and 2 and/or $S^-$ and $S^+$).
The recurrence for the running time remains the same.

%%%%%%%%%%%%%%%%%%%%

\begin{theorem}\label{thm:tri:2contact}
Given a triangle $P$ and a convex $n$-gon $Q$,
we can find the largest similar copy of $P$ contained in $Q$ that has $2$ vertices of $P$ at $2$ vertices of $Q$ in $\OO(n)$ time.
\end{theorem}

\subsection{Algorithm for $k>3$}

We can also adapt the algorithm in Theorem~\ref{thm:kgon:3contact} to handle the 2-contact case
for general $k$.
We use a variant of Lemma~\ref{lem:kgon:rs} where  $\varphi(p_1)$ is a vertex $v_1$ of $\Gamma_1$, $\varphi(p_2)$ is a vertex $v_2$ on $\INT(v_1)\subseteq \Gamma_2$, and
for each $i\in\{3,\ldots,k\}$, $\varphi(p_i)$ is left of $\ext{\Gamma_k}$.
The disjointness conditions for the arcs are the same.
The proof is very similar to that of Lemma~\ref{lem:kgon:rs}; however, for each vertex $v_2$ on a sub-edge $e_2$ of $\Gamma_2$,
we redefine 
\[ \begin{array}{rrl}
\PPP(e_2)\ = & \{(s,t,u,v)\in\R^4:  & \mbox{$\varphi_{s,t,u,v}(p_2) = v_2$ and $\varphi_{s,t,u,v}(p_i)$ is left of $\ext{\MATCH_{i}(e_2)}$ for all } \\[.5ex]
 &&\mbox{$i\in\{3,\ldots,k\}$}\}\\[1ex]
 R_\rho(e_2)\ = & \{\varphi_{s,t,u,v}(p_1): & \mbox{$(s,t,u,v)\in\PPP(e_2)$ and $s^2+t^2\ge\rho^2$}\}.
\end{array}
\]
Similar to before, $\PPP(v_2)$ is a 2-dimensional polygon in $\R^4$ with $O(k)$ edges (since there are 2 linear equality constraints and $O(k)$
inequality constraints in 4 variables), and $R_\rho(v_2)$ is a region in $\R^2$ which is the intersection of a convex $O(k)$-gon and the
exterior of an ellipse (actually, a circle in this case). We partition the region into $O(k)$ constant complexity regions as before.

The rest of the divide-and-conquer algorithm in Theorem~\ref{thm:kgon:3contact} requires no major changes.

\begin{theorem}\label{thm:kgon:2contact}
Given a $k$-gon $P$ and a convex $n$-gon $Q$ (where $k \le n$),
we can find the largest similar copy of $P$ contained in $Q$ that has $2$ vertices of $P$ at $2$ vertices of $Q$
in $O(k^{O(1/\eps)}n^{1+\eps})$ time for any $\eps>0$.  % $n2^{O(\sqrt{\log k\log n})}$ time.
\end{theorem}

}%%%%%%%%%%%%%%%%%

\end{document}